\documentclass[twoside,11pt]{article}

\usepackage{obs_study_style}

\usepackage{amsfonts,amsmath,color,subfigure}
\usepackage{bbm}
\usepackage{algorithm,enumerate}
\usepackage{multirow}
\usepackage{algpseudocode}
\usepackage{setspace}
\usepackage{paralist}

\newcommand{\mbf}{\mathbf}
\newcommand{\sbf}{\boldsymbol}

\def\n{\nonumber}                   

\def\supw{\sup\limits_{ \w\in \calH_J}}

\def\w{{\sbf w}}
\def\tr{\hbox{trace}}      
\def\calC{{\cal C}}        
          
\def\calH{{\cal H}}          
\def\0{{\bf 0}}
\def\A{{\bf A}}

\def\I{{\bf I}}
\def\e{{\sbf \epsilon}}
\def\y{{\bf y}}

\newcommand{\1}{\mathbf{1}}

\newcommand{\bmu}{\boldsymbol{\mu}}

\def\be{\begin{eqnarray}}
\def\ee{\end{eqnarray}}
\def\bSig{{\bf \Sigma}}

\newcommand{\E}[1]{\mathbb{E} \left[#1\right]}

\heading{}{}{}{}{}{Rong J.B. Zhu}

\ShortHeadings{Synthetic Regressing Control}{Zhu}
\firstpageno{1}

\begin{document}

\title{Synthetic Regressing Control}

\author{\name Rong J.B. Zhu \email rongzhu@fudan.edu.cn \\
        \addr Institute of Science and Technology for Brain-Inspired Intelligence\\
        Fudan University \\
        Shanghai, China}

\maketitle

\begin{abstract}
Estimating weights in the synthetic control method, typically resulting in sparse weights where only a few control units have non-zero weights, involves an optimization procedure that selects and combines control units to closely match the treated unit.  
However, it is not uncommon for the linear combination of pre-treatment period outcomes for the control units, using nonnegative weights with the constraint that their sum equals one, to inadequately approximate the pre-treatment outcomes for the treated unit. 
To address the issue, this paper proposes a simple and effective method called \emph{Synthetic Regressing Control} (SRC).  The SRC method begins by performing the univariate linear regression to appropriately align the pre-treatment periods of the control units with the treated unit. Subsequently, a SRC estimator is obtained by synthesizing the regressed controls. To determine the weights in the synthesis procedure, we propose an approach that utilizes a criterion of an unbiased risk estimator. Theoretically, we show that the synthesis way is asymptotically optimal in the sense of achieving the minimum loss of the infeasible best possible synthetic estimator.  Extensive numerical experiments highlight the advantages of the SRC method.
\end{abstract}


\section{Introduction}
\label{sec:intro}

The synthetic control (SC) method is a popular approach of evaluating the effects of policy changes. It allows estimation of the impact of a treatment on a single unit in panel data settings with a modest number of control units and with many pre-treatment periods (\citealp{AG:03}, and \citealp{ADH:10}). The key idea under the SC method is to construct a weighted average of control units, known as a synthetic control, that matches the treated unit’s pre-treatment outcomes. The estimated impact is then calculated as the difference in post-treatment outcomes between the treated unit and the synthetic control.  See \cite{Abadie:21} for recent reviews.

The SC method utilizes constrained optimization to solve for weights, typically resulting in sparse weights where only a few control units have non-zero weights (\citealp{AL:21}). 
This procedure selects and combines control units to closely match the treated unit, enabling extrapolation. 
However, the synthetic control method may fail to extrapolate well. 
For instance, under factor models, if the factor loading of the treated unit lies outside the convex hull of the control units’ factor loadings,
then no weight vector on the simplex can balance the factor loadings.
Formally, we identify two key oversights of the method.
Firstly, interpolation error may occur when the pre-treatment outcomes for the control units do not closely match those for the treated unit. 
Secondly, the constraint that the sum of weights equals 1 overlooks the impact of noise. 
These two oversights lead to suboptimal weights.

In this article, we present a straightforward yet effective method named \emph{Synthetic Regressing Control} (SRC)  to address these oversights, 
retaining the strengths of the synthetic control strategy while reducing interpolation error. 
The SRC procedure begins by fitting each control unit to the treated unit using pre-treatment data. Through univariate linear regressions, SRC refines this alignment during the pre-treatment period, enhancing the predictive performance of control units. This improved fit substantially reduces interpolation error, which is a key advantage of the SRC approach.

The method then computes a weighted average of these pre-treatment fits across all control units to generate an SRC estimator.  
To determine the weights in the synthesis procedure, we propose an approach that utilizes a criterion of unbiased risk estimator. 
The criterion is a penalized loss minimization, equivalent to the loss minimization in the SC method but with the constraint that the sum of nonnegative weights is bounded by a constant, which accounts for model noise. 
SRC weights control units based on the goodness of fit of their pre-treatment data to the treated unit. Units with a good fit receive non-zero weights, signifying a substantial reliance on their data. Conversely, a zero weight indicates the corresponding unit is not heavily utilized by SRC, likely due to a poor pre-treatment fit.

Theoretically, we show that the synthesis approach is asymptotically optimal in the sense of achieving the minimum loss of the infeasible best possible synthetic estimator.  
We conduct both detailed simulation studies and an empirical study of the economic costs of conflicts in Basque, Spain, to shed light on when the SRC method performs well. 
We find evidence that SRC has lower mean-squared prediction error than alternatives in these studies.

The article is organized as follows. Section \ref{sec:related} briefly reviews related work. 
Section \ref{sec:SCM} introduces the SC method and identifies oversights associated with it. Section \ref{sec:SRC} presents the SRC method, which includes Section \ref{sec:Model} on the proposed working model framework, Section \ref{sec:rational} on the rationale for the unit regression step, Section \ref{sec:SRC-sub} on the synthesis method, 
and Section \ref{sec:weight} on an asymptotically optimal weighting method for the SRC method.  
Two extensions are considered: the case when units are more than time periods in Section \ref{sec:HD-setting}, and the incorporation of auxiliary covariates in Section \ref{sec:auxiliary}.  
Section \ref{sec:empirical} reports on extensive simulation studies as well as an application to the Basque dataset. Finally, Section \ref{sec:discussion} discusses the limitations of the method and some possible directions for further research.

\section{Related Work}
\label{sec:related}

This article is closely related to the studies that investigate the SC estimator when the pre-treatment fit is imperfect. 
\cite{FP:21} show the impact of the intercept term for the SC method and argues that a demeaned version of the SC method is already efficient. 
Our approach further investigates the influence of the regression coefficients from the regression of the treated unit on each control unit.

Another approach is to use an outcome model for reducing the imperfect fit. \cite{Powell:18} allows for extrapolation by constructing the SC unit based on the fitted values on unit-specific time periods. 
\citet{Ben:21} propose the augmented synthetic control method, which uses an outcome model to estimate bias resulting from imperfect pre-treatment fit and de-biases the original SC estimator.
The SRC method relates these estimators in the sense of addressing the issue of imperfect pre-treatment fit. However, it differs from them in intent. 
The concept of SRC aims to mitigate interpolation bias through unit regressing 
while concurrently synthesizing all regressed control units by optimizing an unbiased risk estimator criterion.

Our study also relates to literature that relaxes the restriction that the weights are nonnegative. 
\cite{DI:16} argue that negative weights would be beneficial in many settings and proposes adding an intercept into the SC problem. Similarly, \cite{ASS:18} propose a denoising algorithm that combines negative weights with a preprocessing step. 
In contrast, we argue that requiring the sum of weights to equal 1 may lead to suboptimal weights, and instead propose an optimization criterion that constrains the sum of nonnegative weights in a data-driven manner.

Several related articles have addressed the challenge of dealing with datasets that include too many control units, leading to the solution of the SC estimator being not unique (\citealp{ADH:15}).
\cite{RSK:17} and \cite{AL:21} adapt the original SC proposal to incorporate a penalty on the weights into the SC optimization problem. \cite{GM:16} make use of dimension reduction strategies to improve the estimator's performance. 
\cite{DI:16} suggest selecting the set of best controls by restricting the number of controls allowed to be different from zero using an $l_0$-penalty on the weights, 
and \cite{PX:22} further investigate it. 
While the SRC method does not employ a penalty to tackle the problem of an excessive number of units, the criterion used in a penalty-style fashion stems from constructing an unbiased estimator of the risk associated with the synthetic estimator. 
In situations where the number of units is large with respect to the number of time periods, a preprocessing step of screening units can extend the SRC method.

Our article is related to \cite{Kellogg:21}, which propose the matching and synthetic control (MASC) estimator by using a weighted average of the SC and matching estimators to balance interpolation and extrapolation bias. Our SRC method differs from MASC in several ways. First, SRC does not use the matching estimator and instead considers pre-treatment fit of each control unit that aligns the treated unit. 
Secondly, while the MASC estimator combines the SC estimator and the matching estimator through a weighted average, the SRC estimator synthesizes all regressed controls using a weighted average. 
Finally, the methods also differ in terms of how the weights are chosen. For the MASC estimator, the only weight is chosen by the cross-validation method. In contrast,  the SRC method involves solving multiple weights, and we employ the unbiased risk estimator criterion to determine these weights.

Our article is also connected to \cite{Athey:19} and \cite{VB:22}, which have explored the advantage of model averaging within the realm of synthetic control. While \cite{Athey:19} combine several regularized SC and matrix completion estimators developed in \cite{DI:16} and \cite{Athey:21},  \cite{VB:22} combine a large number of estimators from the machine learning literature. 
In contrast, the weighting scheme in our SRC method aims to ensemble all fitted units, mitigating the risk associated with the resulting synthetic estimator. 
We leverage the pre-treatment fitted estimators from each control unit as the proxied control units for synthesis.  
The averaging process in our approach serves to alleviate the inherent risk in the synthetic estimator. 
This distinguishes our method from previous approaches, which involve averaging various types of estimators or combining multiple estimators from the machine learning literature.

In addition to SC-style weighting strategies, there have been articles that directly use outcome modeling approaches. 
These include the panel data approach in \cite{Hsiao:12}, the generalized synthetic control method in \cite{Xu:17}, the matrix completion method in \cite{Athey:21}, and 
the synthetic difference-in-differences method in \cite{Ark:21}. 
In this article, we focus on the synthetic control framework, aiming to retain the accurate extrapolation property introduced in Section \ref{sec:intro}.

\section{Synthetic Control Method}
\label{sec:SCM}

\subsection{Overview}
We consider the canonical SC panel data setting with $j=1,\cdots,J+1$ units observed for $t=1,\cdots, T$ time periods. We restrict attention to the case where a single unit receives treatment, and follow the convention that the first one $j=1$ is treated and that the remaining $J$ ones are control units. Let $T_0$ be the number of pre-intervention periods, with $1\leq T_0<T$.
Let $\mathcal{T}_0$ and $\mathcal{T}_1$ be the set of time indices in the periods of pretreatment and post-treatment, respectively. 
We adopt the potential outcomes framework \citep{Neyman1923};
the potential outcomes for unit $j$ in period $t$ under control and treatment are $Y_{jt}(0)$ and $Y_{jt}(1)$, respectively. 
Thus, the observed outcomes are defined as follows: 
\begin{align*}
    Y_{1t} & = 
    \begin{cases}
    Y_{1t}(0) \text{ if } t\leq T_0,\\
    Y_{1t}(1) \text{ if } t> T_0;
    \end{cases}\\
    Y_{jt} & = Y_{jt}(0) \text{ for } j=2,\cdots, J+1, t=1,\cdots,T.
\end{align*}

We now give assumptions on the underlying Data Generation Process (DGP) for the treated potential outcomes, 
which are divided into a model component $\mu_{1t}(0)$ plus an additive noise term $\epsilon_{1t}$, as given by 
\begin{align}\label{model-po}
Y_{1t}(0)=\mu_{1t}(0)+\epsilon_{1t}, \ \ t\in \{1,\cdots, T\}, 
\end{align}
where $\epsilon_{1t}$ denotes idiosyncratic errors with $\E{\epsilon_{1t}}=0$ and $\E{\epsilon_{1t}^2}=\sigma^2$. 
Note that we assume homoskedasticity of the errors on $\epsilon_{1t}$ for convenience, as we allow for approximation errors on $\mu_{1t}(0)$ in our working model framework, as specified in \eqref{model:mu1-model} in Section \ref{sec:SRC}. 
In this paper, we do not make any assumptions about the control potential outcomes, as they can be used as predictors in the synthetic control method.

We define the effect of the intervention for the treated unit at time $t$ as 
$$\tau_{1t}= Y_{1t}(1) - \mu_{1t}(0), \ \ t\in \{1,\cdots, T\}. $$ 
Different from the treatment effect $Y_{1t}(1)-Y_{1t}(0)$ in previous studies \citep{Abadie:21}, 
here we focus on $Y_{1t}(1) - \mu_{1t}(0)$, which removes the noise component under model \eqref{model-po}. 
Thus, the aim of synthetic controls is to predict $\mu_{1t}(0)$.

Let $\mbf{y}_1$ be a $(T_0\times 1)$ vector of pre-intervention characteristics of the treated unit that we aim to match as closely as possible, and $\mbf{Y}_0$ be $(T_0\times J)$ matrix that contains the same variables for the control units.
A synthetic control is defined as a weighted average of the control units. 
Let $\w=(w_2,\cdots,w_{J+1})^{\top}$ be the
weight vector in the unit simplex in $\mathbb{R}^J$:
\begin{equation*}
  \calH_{\text{sc}}=\left\{w_j\in [0,1]:
    \sum_{j=2}^{J+1} w_j=1\right\}.
\end{equation*}

In the SC method, the weight vector $\w$ is chosen to solve the following optimization problem: 
\begin{align}\label{eq:sc-opt}
\tilde{\w}^{\text{sc}} & =\arg\min_{\w\in\calH_{\text{sc}}} \|\mbf{y}_1-\mbf{Y}_0\w\|.
\end{align}
Then, a synthetic control estimator is constructed by 
\begin{align*}
  \hat{\mu}_{1t}^{\text{sc}}(0) & =
\sum_{j=2}^{J+1}\tilde{w}_j^{\text{sc}}Y_{jt}, \ \ t\in \{1,\cdots, T\}, 
\end{align*}
and the treatment effect $\tau_{1t}$ is estimated by 
\begin{equation*}
  \hat{\tau}_{1t}^{\text{sc}}=Y_{1t}-\hat{\mu}_{1t}^{\text{sc}}(0) = Y_{1t}-\sum_{j=2}^{J+1}\tilde{w}_j^{\text{sc}}Y_{jt}, 
  \ \  t\in \{1,\cdots, T\}. 
\end{equation*}
The weights $\tilde{\w}^{\text{sc}}$ in the SC estimator are typically sparse,  meaning that they are only non-zero for a few control units (\citealp{AL:21}).  This feature is considered as an attractive property since it provides a way for experts to use their knowledge to evaluate the plausibility of the resulting estimates (\citealp{Abadie:21}).

In the SC method, the optimization problem \eqref{eq:sc-opt} involves pursuing of the procedure of synthesizing control units so that the synthetic control is close to the treated unit.
Furthermore, the unit simplex ensures that the weights in $\tilde{\w}^{\text{sc}}$ sum up to 1, representing  a weighted average of the selected control units.

\subsection{Oversights in the SC Method}
\label{sec:issues}

For the synthetic control method, \cite{ADH:10} demonstrate that under certain conditions, there exists a $\w \in \mathcal{H}_{\text{sc}}$ such that the combination of controls is unbiased. 
This means that the synthetic control estimator can be obtained by assuming the following working model: 
\begin{align}
    \mu_{1t}(0) &=\mu_1 + \sum\nolimits_{j=2}^{J+1}w_jY_{jt}+e_{t}, \ \ t\in \{1,\cdots, T\}, 
\label{model:sc-model}
\end{align}
where  $\w\in \mathcal{H}_{\text{sc}}$ denotes the “true” weights, and the approximation errors $(e_1,\cdots, e_T)$ is a sequence with independent, zero-mean variables.  
We include the intercept $\mu_1$ in the model by following the observation in \cite{FP:21}.

Denote $\bar{y}_j=T_0^{-1}\sum_{t=1}^{T_0}Y_{jt}$, $\bar{\mu}_{1}=T_0^{-1}\sum_{t=1}^{T_0}\mu_{1t}(0)$, and $\bar{\epsilon}_1=T_0^{-1}\sum_{t=1}^{T_0}\epsilon_{1t}$.  
A demeaned synthetic control estimator is given as 
$$\tilde{Y}_{1t}(\sbf{w})=\bar{y}_1+\sum\nolimits_{j=2}^{J+1}w_j(Y_{jt}-\bar{y}_j),$$ 
where weights $\sbf{w}$ are chosen to minimize $\sum_{t=1}^{T_0}[\tilde{Y}_{1t}(\sbf{w})-Y_{1t}]^2$ over the set $\mathcal{H}_{\text{sc}}$.

Now, let us decompose the error between $\tilde{Y}_{1t}(\sbf{w})$ and $\mu_{1t}(0)$ as follows: 
\begin{align}\label{eqn:error1}
    \tilde{Y}_{1t}(\sbf{w}) - \mu_{1t}(0) = & 
    \bar{y}_1+\sum\nolimits_{j=2}^{J+1}w_j(Y_{jt}-\bar{y}_j) - \mu_{1t}(0) \notag\\
    =  & \sum\nolimits_{j=2}^{J+1}w_j\left(Y_{jt}- \bar{y}_j - Y_{1t} + \bar{y}_1\right) + \left(\sum\nolimits_{j=2}^{J+1}w_j -1\right)(\mu_{1t}(0)-\bar{\mu}_1) \notag\\
     & +\left(\sum\nolimits_{j=2}^{J+1}w_j\right)\epsilon_{1t} +\left(1- \sum\nolimits_{j=2}^{J+1}w_j\right)\bar{\epsilon}_{1}.     
\end{align}
Unlike the decomposition of interpolation and extrapolation errors presented in \cite{Kellogg:21}, 
 \eqref{eqn:error1} breaks down the error into three components: 
\begin{compactitem}
\item The first term, $\sum\nolimits_{j=2}^{J+1}w_j\left(Y_{jt}- \bar{y}_j - Y_{1t} + \bar{y}_1\right)$, represents the weighted prediction error, which arises from from predicting $Y_{1t}$ using each control unit.
\item The second term, $\left(\sum\nolimits_{j=2}^{J+1}w_j -1\right)(\mu_{1t}(0)-\bar{\mu}_1)$, 
depends on  the gap between $\mu_{1t}$ and $\bar{\mu}_1$, where $\bar{\mu}_{1}$ refers to the mean during the pre-intervention periods for the treated unit. 
\item The third term, $\left(\sum\nolimits_{j=2}^{J+1}w_j\right)\epsilon_{1t} +\left(1- \sum\nolimits_{j=2}^{J+1}w_j\right)\bar{\epsilon}_{1}$, represents the impact of noise. 
\end{compactitem}

From the decomposition \eqref{eqn:error1}, we have two observations.
First, the first term relies on the prediction errors $Y_{jt}- \bar{y}_j - Y_{1t} + \bar{y}_1$ with weights $w_j$.  If they are large, the term may be uncontrollably large, and the weights may not effectively mitigate it. 

Second, the constraint $\sum_{j=2}^{J+1}w_j=1$ in the SC method is specifically designed to minimize the second term.  
When the constraint holds, the second term $(\sum_{j=2}^{J+1}w_j-1)(\mu_{1t}(0)-\bar{\mu}_1)$ disappears. 
However, this constraint overlooks the third term and does not guarantee that it remains small. 
Denote $s_w=\sum\nolimits_{j=2}^{J+1}w_j$.
A simple calculation shows 
$$\mathbb{E}[(s_w\epsilon_{1t} +(1- s_w)\bar{\epsilon}_{1})^2] = s_w^2\sigma^2+(1- s_w)^2T_0^{-1}\sigma^2.$$  
It is minimized at $s_w=1/(T_0+1)$, giving a minimum value of $\sigma^2/(T_0+1)$.
For comparison, it equals $\sigma^2/T_0$ when $s_w = 0$ and $\sigma^2$ when $s_w = 1$. 
Therefore, under the constraint $s_w=1$, the third error is not optimized. 

Consequently, combining the two observations above suggests that minimizing the loss in the SC method under the constraint $\sum_{j=2}^{J+1}w_j=1$ leads to suboptimal weights, as it overlooks the influence of the prediction errors and ignores the impact of the errors $\epsilon_{1t}$.

We conclude this section by noting that the decomposition in \eqref{eqn:error1} does not capture the advantage of imposing nonnegativity constraints on the weights, which is a key factor underlying the superior performance of the SC method relative to OLS. In other words, our analysis is not intended to suggest that the OLS method is optimal. Our approach to choosing weights preserves this key advantage of synthetic control -- the nonnegativity constraints on the weights -- and we introduce it in Section \ref{sec:weight}, where the optimality results explicitly allow for the possibility that the weight vector $\w$ outside the set $\mathcal{H}_{\text{sc}}$ may achieve lower risk. At the same time, constraining $\w$ to be nonnegative is often necessary in practice for  transparency and interpretability, as it ensures that the synthetic unit is a weighted sum of the donor units without any extrapolation in the negative direction -- an important property, particularly when the number of donor units $J$ is large. 
For further discussion on the role of the nonnegativity constraint, see \citet{AL:21}, \citet{Abadie:21}, and  \citet{Kellogg:21}.

\section{Synthetic Regressing Control}
\label{sec:SRC}
In this section, we first introduce a model framework that outlines our approach, and then employ two distinct techniques to address each of the aforementioned oversights separately. 
Specifically, we utilize unit regression to mitigate the interpolation error 
and implement a risk measure to minimize, determining the weight vector $\w$ to reduce the extrapolation error.

\subsection{A Model Framework}
\label{sec:Model}

We define 
\begin{equation}\label{unit-regression-coef}
\theta_j^* = \arg\min_{\theta_j \in \mathbb{R}}\left[\lim\limits_{T\rightarrow\infty}T^{-1}\sum\nolimits_{t=1}^T[Y_{1t}-\mu_1-\theta_j(Y_{jt}-\mu_j)]^2\right],
\end{equation}
where $\mu_j=\lim\limits_{T\rightarrow\infty}T^{-1}\sum\nolimits_{t=1}^TY_{jt}$ for $j=1,\cdots,J+1$.
\eqref{unit-regression-coef} represents a linear regression for each unit $j$, which we refer to as \emph{Unit Regressing}. 
The goal of unit regressing is to establish a correspondence between each control unit and the treated unit, aiming to minimize the distance between them. This process entails conducting a univariate regression analysis where the treated unit is regressed on the control unit. By doing so, we can estimate the counterfactual outcome for each control unit based on the pre-intervention regression fit.  Using the regression method is to mimic the behavior of the treated unit before the intervention as closely as possible.

From the unit regression, we construct the true regressed control values for each $j$ given by: 
\begin{equation}\label{unit-regression-fit}
Y_{jt}^*(0) = \theta_j^*(Y_{jt}-\mu_j),  \ \ t\in \{1,\cdots, T\}.
\end{equation}
Let
$\sbf{w} = (w_2,\cdots,w_{J+1})^{\top}$ be the
weight vector in the set 
\begin{equation*}
  \mathcal{H}_J=\left\{w_j\in [0,1]: j=2,\cdots,J+1\right\}. 
\end{equation*}
We propose our working model framework 
\begin{align}
\mu_{1t} &= \mu_1 + \sum\nolimits_{j=2}^{J+1}w_jY_{jt}^*(0)+e_{t},  \ \ t\in \{1,\cdots, T\}, 
\label{model:mu1-model}
\end{align}
where $\sbf{w}\in \mathcal{H}_J$ denotes the ``true" weights, and the approximation errors $\{e_1,\cdots, e_T\}$ form a sequence with independent, zero-mean variables.

\subsection{Rationale for the Unit Regression Step}
\label{sec:rational}

In the working model framework \eqref{model:mu1-model}, $\sum\nolimits_{j=2}^{J+1}w_jY_{jt}^*(0)$ synthesizes the regressed controls $Y_{jt}^*(0)=\theta_j^*(Y_{jt}-\mu_j)$ using weights $w_j$. 
Compared to the original synthetic control model \eqref{model:sc-model}, the working model \eqref{model:mu1-model} modifies the approach by combining linear transformations of the controls rather than the raw control variables.
Similar to the decomposition in \eqref{eqn:error1}, we have
\begin{align}\label{eqn:error2}
\sum\nolimits_{j=2}^{J+1}w_jY_{jt}^*(0)  - \mu_{1t}
= & \sum\nolimits_{j=2}^{J+1}w_j\left[Y_{jt}^*(0) - (Y_{1t}(0)-\mu_{1})\right] \notag\\
 & + \left(\sum\nolimits_{j=2}^{J+1}w_j - 1\right)(\mu_{1t} - \mu_{1})  +  \sum\nolimits_{j=2}^{J+1}w_j\epsilon_{1t}  - \mu_1. 
\end{align}
The effect of unit regression is reflected in the first term on the right-hand side of \eqref{eqn:error2}.
We next analyze its mean squared error:
$$\E{\left[\sum\nolimits_{j=2}^{J+1}w_j(Y_{jt}^*(0) - (Y_{1t}(0)-\mu_{1}))\right]^2}\leq \sum\nolimits_{j=2}^{J+1}w_j^2\sum\nolimits_{j=2}^{J+1}\E{\left(Y_{1t}(0)  - \mu_1 - Y_{jt}^*(0) \right)^2}.$$ 
It follows that $\E{(Y_{1t}(0) - \mu_1 - Y_{jt}^*(0))^2}$, for each control unit $j\in\{2,\cdots,J+1\}$, contributes to controlling the overall mean squared error.  
To compare our approach with the original synthetic control method, we examine how $\E{(Y_{1t}(0) - \mu_1 - Y_{jt}^*(0))^2}$ differs from $\E{[Y_{1t}(0) - \mu_1 - (Y_{jt}(0)-\mu_j)]^2}$.

We explore three DGP examples -- a linear factor model, a nonlinear factor model, and an autoregressive model -- to show the applicability of our working model framework \eqref{model:mu1-model}, and then substantiate the rationale behind unit regression.

\textbf{Example 1: Linear Factor Model.}
The potential outcomes of unit $j\in \{1,\cdots, J+1\}$ at time $t\in \{1,\cdots, T\}$ are given as follows: 
\begin{align}
Y_{jt}(0) = &\sbf{\lambda}_t^{\top}\mbf{f}_j+\epsilon_{jt}, 
\label{model-po0}
\end{align}
where  
$\sbf{\lambda}_t=(\lambda_{t1}, \cdots, \lambda_{tp})^{\top}$ is a $p\times 1$ vector of unobserved common stochastic factors, with $\E{\sbf{\lambda}_t}=\mbf{b}$,  
 $\E{(\sbf{\lambda}_t-\mbf{b})(\sbf{\lambda}_t-\mbf{b})^{\top}}=\mbf{I}_p$,
and $\mbf{f}_j$ is a $p\times 1$ vector of unknown, fixed factor loadings.  
In this example, $\mu_j=\mbf{b}^{\top}\mbf{f}_j$ for $j=1,\cdots,J+1$. 

Now, we demonstrate the benefit of applying unit regressing. 
A simple calculation shows that $\theta_j^*=\mbf{f}_1^{\top}\mbf{f}_j/(\mbf{f}_j^{\top}\mbf{f}_j+\sigma^2)$. 
Under the model \eqref{model-po0}, unit regression \eqref{unit-regression-fit} follows 
\begin{equation}\label{fit-mse}
\E{(Y_{1t}(0) - \mu_1 - Y_{jt}^*(0))^2} = \mbf{f}_1^{\top}\mbf{f}_1-2\theta_j^*\mbf{f}_1^{\top}\mbf{f}_j+\theta_j^{*2}\mbf{f}_j^{\top}\mbf{f}_j+(1+\theta_j^{*2})\sigma^2.   
\end{equation} 
While for the original control, we have 
\begin{equation}\label{o-mse}
\E{[Y_{1t}(0) - \mu_1 - (Y_{jt}(0)-\mu_j)]^2}= \mbf{f}_1^{\top}\mbf{f}_1-2\mbf{f}_1^{\top}\mbf{f}_j+\mbf{f}_j^{\top}\mbf{f}_j+2\sigma^2.
\end{equation}
Comparing \eqref{fit-mse} with \eqref{o-mse}, we have 
\begin{align*}
& \E{[Y_{1t}(0) - \mu_1 - (Y_{jt}(0)-\mu_j)]^2} - \E{(Y_{1t}(0) - \mu_1 - Y_{jt}^*(0))^2}\notag\\
 = & \left(\mbf{f}_1^{\top}\mbf{f}_j - (\mbf{f}_j^{\top}\mbf{f}_j+\sigma^2)\right)^2/(\mbf{f}_j^{\top}\mbf{f}_j+\sigma^2) + (1-\theta_j^{*2})\sigma^2.
\end{align*}
It demonstrates that in linear factor model, the unit regression step aims to bring each regressed control unit 
closer to the treated unit than its original control. 
This preliminary step suggests that synthesizing the regressed controls, rather than the original controls, 
may lead to enhanced performance. 
The degree of improvement is anticipated to increase as the gap between $\mbf{f}_1^{\top}\mbf{f}_j$ and  $\mbf{f}_j^{\top}\mbf{f}_j$ widens, 
without considering the impact of $\sigma^2$. 

We next show applicability of our working model \eqref{model:mu1-model}. 
Under the model \eqref{model-po0}, we obtain 
\begin{align*}
e_t = & \mu_{1t}(0) - \mu_1 - \sum\nolimits_{j=2}^{J+1}w_j\theta_j^*Y_{jt}(0) 
 = (\sbf{\lambda}_t - \mbf{b})^{\top}\left(\mbf{f}_1 - \sum\nolimits_{j=2}^{J+1}w_j\theta_j^*\mbf{f}_j\right) - \sum\nolimits_{j=2}^{J+1}w_j\theta_j^*\epsilon_{jt}\notag\\
 = & (\sbf{\lambda}_t - \mbf{b})^{\top}\left(\mbf{I} -\sum\nolimits_{j=2}^{J+1} w_j\mbf{f}_j\mbf{f}_j^{\top}/(\mbf{f}_j^{\top}\mbf{f}_j+\sigma^2)\right)\mbf{f}_1 - \sum\nolimits_{j=2}^{J+1}w_j\theta_j^*\epsilon_{jt}.
\end{align*}
Clearly, the working model framework \eqref{model:mu1-model} holds as $\mathbb{E}[e_t]=0$. 
The error includes two terms: one relies on $\mbf{I} -\sum\nolimits_{j=2}^{J+1} w_j\mbf{f}_j\mbf{f}_j^{\top}/(\mbf{f}_j^{\top}\mbf{f}_j+\sigma^2)$; 
the other relies on the weighted average on $\{\epsilon_{jt}\}_{j=2}^{J+1}$. 
The working model framework provides a basis for determining the optimal $w_j$ that minimizes the approximation errors in some sense, as discussed in the next Section \ref{sec:weight}.

\textbf{Example 2: Nonlinear Factor Model.}
The potential outcomes of unit $j\in \{1,\cdots, J+1\}$ at time $t\in \{1,\cdots, T\}$ are given as follows: 
\begin{align}\label{model-po-non}
Y_{1t}(0) = &\sbf{\lambda}_t^{\top}\mbf{f}_1+\epsilon_{1t} \text { and } 
Y_{jt}(0) = \sbf{\lambda}_{2t}^{\top}\mbf{f}_j+\epsilon_{jt} \text{ for } j=2, \cdots, J+1,
\end{align}
where we adapt the notations in \eqref{model-po0} but denote 
$\sbf{\lambda}_{2t}=(\lambda_{t1}^2, \cdots, \lambda_{tp}^2)^{\top}$, following 
$$\E{(\sbf{\lambda}_{2t}-\mbf{b}_{2})(\sbf{\lambda}_t-\mbf{b})^{\top}}=\text{diag}\{\mbf{0}_p\} \ \text{and } \E{(\sbf{\lambda}_{2t}-\mbf{b}_{2})(\sbf{\lambda}_{2t}-\mbf{b}_{2})^{\top}}=3\mbf{I}_p,$$
where $\mbf{b}_{2}=\E{\sbf{\lambda}_{2t}}$. 
In this example, $\mu_1=\mbf{b}^{\top}\mbf{f}_1$ and $\mu_j=\mbf{b}_{2}^{\top}\mbf{f}_j$ for $j=2,\cdots,J+1$. 
Clearly, in this nonlinear factor model, the control units do not contribute to explaining the treated unit, indicating that the synthetic control method is not applicable.  

We now show that unit regressing can capture this information. 
Under this model, unit regression \eqref{unit-regression-fit} follows 
$\theta_j^*=0$ and $Y_{jt}^*(0)=0$, indicating that no control should be used for synthetic controls. 
We have 
\begin{equation}\label{fit-mse-2}
\E{(Y_{1t}(0) - \mu_1 - Y_{jt}^*(0))^2} = \mbf{f}_1^{\top}\mbf{f}_1 + \sigma^2.  
\end{equation}
For the difference relative to the original control, we have 
$Y_{1t}(0) -\mu_1 - (Y_{jt}(0)-\mu_j) =   \sbf{\lambda}_t^{\top}\mbf{f}_1 -  \sbf{\lambda}_{2t}^{\top}\mbf{f}_j + \epsilon_{1t}-\epsilon_{jt}$, following 
\begin{align}\label{o-mse-2}
\E{[Y_{1t}(0) -\mu_1 - (Y_{jt}(0)-\mu_j)]^2}=& \mbf{f}_1^{\top}\mbf{f}_1+3\mbf{f}_j^{\top}\mbf{f}_j+2\sigma^2.
\end{align}
Comparing \eqref{fit-mse-2} with \eqref{o-mse-2}, we have 
$$ \E{[Y_{1t}(0) -\mu_1 - (Y_{jt}(0)-\mu_j)]^2} - \E{(Y_{1t}(0) - \mu_1 - Y_{jt}^*(0))^2} = 3\mbf{f}_j^{\top}\mbf{f}_j + \sigma^2.$$
It demonstrates that in this nonlinear factor model, the regressed controls obtained in the unit regression step can identify the control units unrelated to the treated unit, thereby ensuring robust performance and potentially leading to improvements.

We next show applicability of our working model \eqref{model:mu1-model}. 
Under the model \eqref{model-po-non}, we obtain 
\begin{align*}
e_t = \mu_{1t}(0) - \sum\nolimits_{j=2}^{J+1}w_j\theta_j^*Y_{jt}(0) 
 = &(\sbf{\lambda}_t - \mbf{b})^{\top}\mbf{f}_1. 
\end{align*} 
Clearly, the working model framework \eqref{model:mu1-model} holds as $\mathbb{E}[e_t]=0$. 
As expected, the approximation error does not depend on the controls, leading to the optimal $w_j$ being 0 for $j=2, \cdots, J+1$. 
In other words, the synthetic control method is not effective for this model.

\textbf{Example 3: Autoregressive Model.} 
To simplify notation, in this example we limit to one post-treatment observation $T=T_0+1$. 
For time period $t=T_0+1$, the potential outcomes $Y_{jt}(0)$ of unit $j\in \{1,\cdots, J+1\}$ are generated as 
\begin{align}\label{model-po-ar}
Y_{jt}(0)=&\mu_1+\sum\nolimits_{\ell=1}^{T_0}\beta_{\ell}(Y_{j(t-\ell)}(0)-\mu_j)+\epsilon_{jt},
\end{align}

Note that $\theta_j^*=\arg\min_{\theta\in\mathbb{R}}\mathbb{E}\left[Y_{1(T_0+1-\ell)}(0)-\mu_1-\theta (Y_{j(T_0+1-\ell)}(0)-\mu_j)\right]^2$. 
This leads to the decomposition on $Y_{1(t-\ell)}(0)$ as $Y_{1(t-\ell)}(0)-\mu_1=\theta_j^* [Y_{j(t-\ell)}(0) -\mu_j]+ \delta_{j(t-\ell)}$, 
where $\theta_j^* Y_{j(t-\ell)}(0)$ represents the projection of $Y_{1(t-\ell)}(0)$ onto $Y_{j(t-\ell)}(0)$ and  
$\delta_{j(t-\ell)}$ is the residual of the projection. 
It illustrates that in autoregressive model, the unit regression step possesses an attractive property:  
it can effectively reduce the fitting error.

We next show applicability of our working model \eqref{model:mu1-model}. 
Under the model \eqref{model-po-ar}, we obtain 
\begin{align*}
e_t & = \mu_{1t}(0)-\mu_1 - \sum\nolimits_{j=2}^{J+1}w_j\theta_j^*(Y_{jt}(0) - \mu_j) \notag\\
& = \sum\nolimits_{\ell=1}^{T_0}\beta_{\ell}\left[Y_{1(t-\ell)}(0) -\mu_1- \sum\nolimits_{j=2}^{J+1}w_j\theta_j^*(Y_{j(t-\ell)}(0)-\mu_j)\right]\notag\\
& = (1 - \sum\nolimits_{j=2}^{J+1}w_j)\sum\nolimits_{\ell=1}^{T_0}\beta_{\ell}(Y_{1(t-\ell)}(0) - \mu_1)
- \sum\nolimits_{j=2}^{J+1}w_j\sum\nolimits_{\ell=1}^{T_0}\beta_{\ell}\delta_{j(t-\ell)}
\end{align*}
Clearly, we have $\mathbb{E}[e_t]=0$, indicating that the working model framework \eqref{model:mu1-model} holds. 
The approximation error consists of a weighted average of two terms, $\sum\nolimits_{\ell=1}^{T_0}\beta_{\ell}(Y_{1(t-\ell)}(0) - \mu_1)$ and $-\sum\nolimits_{\ell=1}^{T_0}\beta_{\ell}\delta_{j(t-\ell)}$ with weights $1 - \sum\nolimits_{j=2}^{J+1}w_j$ and $\sum\nolimits_{j=2}^{J+1}w_j$, respectively. 
As discussed in Example 1, the working model framework serves as a basis for finding the optimal $w_j$ that minimizes the approximation errors in some sense, as explained in the next Section \ref{sec:weight}.

We conclude the examples by highlighting a limitation of the working model \eqref{model:mu1-model}, which may break down in non-stationary settings.
In such cases, the original synthetic control method, through the unit-sum constraint on the weights, may help offset non-stationarity by anchoring the post-treatment counterfactual to a convex combination of observed control units.
In the factor model examples considered here, time-varying effects are treated as random, so there is no need to rely on the unit-sum constraint on the weights, and the working model remains valid. 
When stationarity does not hold, the weights that are optimal during the pre-treatment period may no longer be optimal in the post-treatment period.   
See the Discussion section for details.

\subsection{Synthetic Regressing Control}
\label{sec:SRC-sub}

Denote $\mbf{y}_j=(Y_{j1},\cdots,Y_{jT_0})^{\top}$.  
We estimate $\theta_j^*$ in \eqref{unit-regression-coef} of unit regressing by solving the following optimization: 
\begin{align*}
\min_{\theta_j}\sum\nolimits_{t=1}^{T_0}\left[Y_{1t}-\bar{y}_1-\theta_j(Y_{jt}-\bar{y}_j)\right]^2.
\end{align*}
It follows the least squares estimators 
\begin{align*}
\hat{\theta}_j = & \frac{(\mbf{y}_j-\bar{y}_j\mbf{1})^{\top}(\mbf{y}_1-\bar{y}_1\mbf{1})}{\|\mbf{y}_j-\bar{y}_j\mbf{1}\|^2}
\end{align*}
for each $j$. 
Consequently, the regressed controls of unit $j$ are given by 
\begin{align}\label{eqn:tilde-control}
\tilde{Y}_{jt}(0) = &\hat{\theta}_j(Y_{jt}(0)-\bar{y}_j), \ \ t\in\{1,\cdots, T\}. 
\end{align}

We utilize the synthetic method by assigning a weight vector for the regressed controls.   
Given that $\mu_1$ is estimating by $\bar{y}_1$, and substituting in $\hat{\theta}_j$,  
a \emph{synthetic regressing control} estimator is formulated as 
\begin{equation}\label{predict-syn}
  \hat{\mbf{y}}_1(\sbf{w}) =
\bar{y}_1\mbf{1}+\sum_{j=2}^{J+1}w_j\hat{\theta}_j(\mbf{y}_j-\bar{y}_j\mbf{1}). 
\end{equation}
We refer to it as \emph{``Synthetic Regressing Control"} (SRC). 
In contrast to SC, the weights in SRC are distinct from the regression coefficients. 
The weights indicate the extent to which the regressed controls are considered in the synthesis process, 
while the regression coefficients capture the relationship between the treated and control units.
Similar to SC, SRC employs the synthesis method to control extrapolation error by assigning more weight to the regressed controls that demonstrate higher similarity.

\subsection{Determination of optimal $\sbf{w}$}
\label{sec:weight}
As discussed in Section \ref{sec:issues}, the SC method lacks optimal weights. 
Here, we provide a method for obtaining optimal weights in terms of minimizing an unbiased estimator of the risk. 
We further demonstrate the weights we derive have asymptotic optimality.

Denote $\tilde{\mbf{y}}_j=(\tilde{Y}_{j1}(0), \cdots, \tilde{Y}_{jT_0}(0))^{\top}$. 
From \eqref{eqn:tilde-control}, 
We rewrite $\tilde{\mbf{y}}_j$ as 
$\tilde{\mbf{y}}_j=\mbf{H}_j\mbf{y}_1$, where $\mbf{H}_j=\mbf{Q}\mbf{y}_j\left(\mbf{y}_j^{\top}\mbf{Q}\mbf{y}_j\right)^{-1}\mbf{y}_j^{\top}\mbf{Q}$ 
and $\mbf{Q}=\mbf{I}-T_0^{-1}\mbf{1}\mbf{1}^{\top}$,  
implying that \eqref{predict-syn} is rewritten as 
$$\hat{\mbf{y}}_1(\sbf{w})=\bar{y}_1\mbf{1}+\sum\nolimits_{j=2}^{J+1}w_j\mbf{H}_j\mbf{y}_1.$$ 
Denote the loss of $\hat{\mbf{y}}_1(\w)$ relative to $\bmu_1$ as $$L(\w)=\|\hat{\mbf{y}}_1(\w)-\bmu_1\|^2$$ and define the risk as $R(\sbf{w})=\E{L(\w) | \mbf{Y}_0}$, where 
the expectation is taken over the only source of randomness, $\mbf{\epsilon}$. 
In this context, $\mbf{Y}_0=(\mbf{y}_2, \cdots, \mbf{y}_{J+1})$ is treated as a design matrix, as discussed in Section \ref{sec:Model}. 
Denote $\mbf{H}(\sbf{w})=\sum\nolimits_{j=2}^{J+1}w_j\mbf{H}_j$ with $\ell_{jt}$ as the $t$-th diagonal element of $\mbf{H}_j$.
We have that 
\begin{align}\label{eq:risk}
\E{\|\hat{\mbf{y}}_1(\sbf{w})-\mbf{y}_1\|^2   | \mbf{Y}_0}-R(\sbf{w})
& = \E{\|\hat{\mbf{y}}_1(\sbf{w})-\mbf{y}_1\|^2-\|\hat{\mbf{y}}_1(\sbf{w})-\sbf{\mu}_1\|^2   | \mbf{Y}_0}\notag\\
& = \E{\mbf{\epsilon}^{\top}\mbf{\epsilon}-2\mbf{\epsilon}^{\top}(\bar{\mu}_1\mbf{1}+\bar{\epsilon}\mbf{1}+\mbf{H}(\sbf{w})\sbf{\mu}_1-\sbf{\mu}_1+\mbf{H}(\sbf{w})\mbf{\epsilon})  | \mbf{Y}_0}\notag\\
& =(1-2T_0^{-1})\sigma^2T_0-2\sigma^2\sum\nolimits_{j=2}^{J+1}w_j\sum\nolimits_{t=1}^{T_0}\ell_{jt}.  
\end{align}
Noting that $\sum\nolimits_{t=1}^{T_0}\ell_{jt}=1$, 
 \eqref{eq:risk} demonstrates that the expression 
$$\|\hat{\mbf{y}}_1(\sbf{w})-\mbf{y}_1\|^2+2\sigma^2\sum\nolimits_{j=2}^{J+1}w_j-(1-2T_0^{-1})\sigma^2T_0$$ 
serves as an unbiased estimator of $R(\sbf{w})$. 
This motivates the utilization of the following criterion to obtain $\w$: 
\begin{equation*}
\calC_0(\sbf{w})=\|\hat{\mbf{y}}_1(\sbf{w})-\mbf{y}_1\|^2+2\sigma^2 \sum\nolimits_{j=2}^{J+1}w_j. 
\end{equation*}
This means that $2\sigma^2 \sum\nolimits_{j=2}^{J+1}w_j$ is an estimate of the gap between the risk and the loss, 
ignoring a constant. 
Minimizing the criterion $\calC_0(\sbf{w})$ is equivalent to solving the following optimization problem:
\begin{equation*}
\min_{\sbf{w}}\|\hat{\mbf{y}}_1(\sbf{w})-\mbf{y}_1\|^2 \ \ \text{subject to }  \sum\nolimits_{j=2}^{J+1}w_j\leq c
\end{equation*}
for a constant $c$ that depends on $\sigma^2$.  
Thus, in our approach, the constraint $\sum\nolimits_{j=2}^{J+1}w_j=1$ is replaced by $\sum\nolimits_{j=2}^{J+1}w_j\leq c$.

Estimating $\sigma^2$ by 
\begin{align}\label{eqn:sigmahat}
\hat{\sigma}^2=&\|\mbf{Q}\mbf{y}_1-\mbf{Q}\mbf{Y}_0[\text{diag}(\mbf{Y}_0^{\top}\mbf{Q}\mbf{Y}_0)]^{-1}\mbf{Y}_0^{\top}\mbf{Q}\mbf{y}_1\|^2, 
\end{align}
where $\text{diag}(\mbf{Y}_0^{\top}\mbf{Q}\mbf{Y}_0)$ denotes the diagonal matrix formed by the diagonal elements of $\mbf{Y}_0^{\top}\mbf{Q}\mbf{Y}_0$, we propose a Mallows' $C_p$ criterion 
\begin{equation}\label{M-criterion}
\calC(\sbf{w})=\|\hat{\mbf{y}}_1(\sbf{w})-\mbf{y}_1\|^2+2 \hat{\sigma}^2\sum\nolimits_{j=2}^{J+1}w_j. 
\end{equation}
From (\ref{M-criterion}), the weight vector is obtained as
\begin{equation*}
  \hat{\sbf{w}}= \mathop{\arg\min}\limits_{\sbf{w} \in\mathcal{H}_J}{\calC}(\sbf{w}).
\end{equation*}
With $\hat{\w}$ into \eqref{predict-syn}, we obtain the SRC estimator
\begin{equation}\label{eq:sSRC-estimator}
 \hat{Y}_{1t}(0)=\bar{y}_1+\sum\nolimits_{j=2}^{J+1}\hat{w}_j\hat{\theta}_j(Y_{jt}-\bar{y}_j), \ \ t\in \{1,\cdots, T\}. 
\end{equation}
We summarize the procedure of obtaining the SRC estimator as Algorithm \ref{alg:main}.

In \eqref{eq:sSRC-estimator}, the SRC estimator is represented  as a linear weighting estimator of the outcomes of control units $Y_{jt}$, similar to the SC estimator. 
The weights $\hat{w}_j\hat{\theta}_j$ can be understood as the adjusted weights within the synthetic control method. 
It is essential to recognize that these weights comprise two components: $\hat{w}_j$ and $\hat{\theta}_j$. 
This formulation allows for negative weights by the unit regression coefficients $\hat{\theta}_j$ and facilitates extrapolation beyond the convex hull of the control units. 
In unit regressing alone, the estimator $\hat{\theta}_j$ permits arbitrarily weights even in the absence of correlation between the treated unit and the control unit $j$. 
In contrast, by imposing the constraint of the convex hull of the regressed control units on $\hat{\sbf{w}}$, the sum of weights is penalized. 
This constraint effectively manages the extent of extrapolation error.

\begin{algorithm}
\begin{algorithmic}
\item (1) Obtain  $\hat{\mbf{y}}_1(\sbf{w})$ for $j\in\{2,\cdots,J+1\}$  from \eqref{predict-syn}.
\item (2) Obtain  $\hat{\sigma}^2$ from \eqref{eqn:sigmahat}. 
\item (3) Solve $\hat{\w}$ by 
\begin{equation*}
  \hat{\w}= \mathop{\arg\min}\nolimits_{\w \in\mathcal{H}_J}\left\{\|\y_1-\hat{\mbf{y}}_1(\sbf{w})\|^2+2\hat{\sigma}^2 \w^{\top}\1\right\}.
\end{equation*}
\item (4) Obtain $\hat{Y}_{1t}(0)$ for $t\in\{1,\cdots,T\}$  from \eqref{eq:sSRC-estimator}.
\end{algorithmic}
\caption{The SRC estimator}
\label{alg:main}
\end{algorithm}

Denote $\bmu_1^{(o)}=(\mu_{1(T_0+1)}, \cdots, \mu_{1T})^{\top}$ and let $\hat{\mbf{y}}_1^{(o)}(\w)$ denote the corresponding predictions. 
We also denote $\mbf{y}_j^{(o)}=(Y_{j(T_0+1)}(0), \cdots, Y_{jT}(0))^{\top}$ for $j=1,\cdots,J+1$, and let $\mbf{Y}_0^{(o)}=(\mbf{y}_2^{(o)}, \cdots, \mbf{y}_{J+1}^{(o)})$. 
Define the loss of predicting post-intervention periods as 
$$L^{(o)}(\w)=\|\hat{\mbf{y}}_1^{(o)}(\w)-\bmu_1^{(o)}\|^2.$$ 
In the following theorem, we establish a property regarding the weight vector $\hat{\w}$, solved based on the unbiased risk estimator criterion: $L^{(o)}(\hat{\w})$ is asymptotical attach the minimum loss of the infeasible best possible synthetic estimator, $\inf_{\w\in \mathcal{H}_J}L^{(o)}(\w)$.
This form of asymptotic optimality is a well-established statistical property in model selection \citep{Li:87} and model averaging \citep{Hansen:07,Wan:10,Zhu:23}. Here, we investigate it in the context of synthetic controls.
For simplifying the notation, we assume, without loss of generality, that $\mbf{y}_j, \forall j=2,\cdots,J+1$ and $\mbf{\mu}_1$ are centered, i.e., 
$\mbf{1}^{\top}\mbf{y}_j=0$ and $\mbf{1}^{\top}\mbf{\mu}_1=0$. 
\begin{theorem}
\label{th:opt-out}
Under Model \eqref{model-po}, assume that (1) $\max_t\E{\epsilon_t^4}\leq c_1<\infty$ for some constant $c_1$,
(2) $\|\bmu_1^{(o)}\|^2/(T-T_0)\leq c_2<\infty$ for some constant $c_2$,  
(3) $J^{-1}(T-T_0)^{-1}T_0\|\bmu_1^{(o)}-\mbf{Y}_0^{(o)}[\text{diag}(\mbf{Y}_0^{\top}\mbf{Y}_0)]^{-1}\mbf{Y}_0^{\top}\bmu_1\|^2\rightarrow_p \infty$ as $T_0\rightarrow \infty$, 
and (4) the model \eqref{model:mu1-model} holds, 
then as $T_0\to\infty $, 
\begin{align*}
\frac{L^{(o)}(\hat{\w})}{\inf_{\w\in \mathcal{H}_J}L^{(o)}(\w)}\rightarrow_p 1. 
\end{align*}
\end{theorem}
The technical proof of the following theorem is given in Appendix A. 
Theorem \ref{th:opt-out} demonstrates the asymptotic optimality of the proposed method for out-of-sample predictions, i.e., predicting post-intervention periods, under the model \eqref{model:mu1-model}.  
The synthetic estimator with the weights $\hat{\w}$ asymptotically achieves the minimum loss of the infeasible best possible synthetic estimator. 
Additionally, we also present the asymptotic optimality for in-sample prediction presented in Theorem \ref{th:opt} of the Appendix A.

The conditions of $\max_t\E{\epsilon_t^4}\leq c_1$ and $\|\bmu_1^{(o)}\|^2/(T-T_0)\leq c_2$ 
are quite mild since they only require
bounded fourth moments of errors and that $\|\bmu_1^{(o)}\|^2=O(T-T_0)$, respectively. 
The crucial condition, 
$$J^{-1}(T-T_0)^{-1}T_0\|\bmu_1^{(o)}-\mbf{Y}_0^{(o)}[\text{diag}(\mbf{Y}_0^{\top}\mbf{Y}_0)]^{-1}\mbf{Y}_0^{\top}\bmu_1\|^2\rightarrow_p \infty$$ as $T_0\rightarrow \infty$, means that the squared post-intervention prediction error is large relative to the number of control units. 
This condition described is typically considered to be mild in the context of the synthetic control problem when $T_0$ is large relative to $J$. This is because achieving a perfect approximation through univariate regression on a simple control unit is rarely attainable. 
When $J\geq T_0$ or $J\approx T_0$, one practical way is to screen the units to reduce the number of units. Further details are discussed in Section \ref{sec:HD-setting}. 
Additionally, when incorporating auxiliary covariates with dimension $d$,  $T_0$ would be extended to $T_0+d$. Further insights on this aspect are provided in Section \ref{sec:auxiliary}.

The key condition of this theorem is that the working model \eqref{model:mu1-model} holds.  
The working model ensures the effectiveness of predicting post-intervention periods. 
While we have demonstrated that the working model holds for the three DGPs in Section \ref{sec:rational}, further verification is needed to assess whether more DGPs violate the working model framework.

Theorem \ref{th:opt-out} shows that the oracle loss relative to the true $\bmu_1^{(o)}$ for the proposed estimator converges to the optimal loss achieved by the oracle weights, which are assumed to be known. However, it does not provide an unbiased estimate of $\bmu_1^{(o)}$, and thus does not yield an unbiased estimator for the treatment effect. Deriving a debiased estimator under model \eqref{model:mu1-model} using the projection theory \citep{ZZ:14,VB:14,Li:20} to enable valid statistical inference is a promising direction for future work, though it lies beyond the scope of this paper.

\section{Extensions}
\label{sec:extension}
In this section, we consider \text{two} elaborations to the basic setup. First, we extend it to cases where units are more than time periods. Second, we extend it by incorporating auxiliary covariates. In addition, a placebo-based permutation test \citep{ADH:10} is presented in Section \ref{sec:test} of the Appendix.

\subsection{Screening Units When They are Too Many}
\label{sec:HD-setting}

We extend the application of the SRC method to cases where $J\geq T_0$ or $J\approx T_0$. 
To accomplish this, we propose a practical procedure that involves screening the units using the sure independent ranking and screening (SIRS) method (\citealp{zhu2011model}) to reduce the number of units. 
In high-dimensional statistics, Theorems 2 and 3 in \cite{zhu2011model} indicate that SIRS can reduce the dimensionality without losing any active variables with a probability approaching one. 
We prefer SIRS over the original sure independence screening proposed by \cite{Fan:08} because it allows us to assume that no linear candidate model is correct.

For applying SIRS into the control units, we assume that $\mu_{1t}$ depends only on some of the control units, called as active units, in this study. 
SIRS screens the units based on the magnitude of the following statistics instead of the marginal correlation,
\begin{align}\label{mag:eqn}
 \tilde{\eta}_j=&\frac{1}{T_0}\sum_{t=1}^{T_0}\left\{\frac{1}{T_0}\sum_{\ell=1}^{T_0}Y_{jt}I_{(-\infty,\ Y_{1t})}(Y_{1\ell})\right\}^2 \quad\text{for}\quad j=2,\cdots,J+1.
 \end{align}
Derivation and interpretation of this statistics can be found in \cite{zhu2011model}.
We use the statistics $\tilde{\eta}_j$ for screening units, then obtain a set that involves any activate units. 
Following \citet{zhu2011model}, we set the number of active units to the nearest integer of $T_0/\log (T_0/2)$. We summarize it as Algorithm \ref{alg:HD} below.
 \begin{algorithm}
 \begin{algorithmic}
 \item  Step 1: Screen units by the SIRS method to get the subset $\mbf{Y}_s$ from $\mbf{Y}_0$.
 \item  \quad Step 1.1:  Calculate the magnitudes $\tilde{\eta}_j$ for $j=2,\cdots, J+1$ according to \eqref{mag:eqn};
 \item  \quad Step 1.2:  Select $k=\lfloor T_0/\log (T_0/2) \rceil $ units among the $J$ control units with the largest $\tilde{\eta}_j$ values;
  \item  \quad Step 1.3:  Construct the subset $\mathbf{Y}_s$ from $\mathbf{Y}_0$ using the selected units.
 \item  Step 2: Perform Algorithm \ref{alg:main} on $\mbf{y}_1$ and $\mbf{Y}_s$. 
 \end{algorithmic}
 \caption{The SRC estimator when control units are too many}
 \label{alg:HD}
 \end{algorithm}
 
Regarding unit screening, in practice we recommend performing screening when $J\geq 4T_0/5$ to ensure sufficient degrees of freedom for computing the Mallows’ $C_p$ criterion. 
Once the screened units are reduced, we perform the SRC method on these units to obtain the estimator.
It is worth noting that the first step of screening units differs from that of \citet{Xu:17} and \citet{ASS:24}, where the initial step involves estimating the latent factors and factor loadings using control units during the pre-treatment period. We also provide an empirical comparison with the generalized synthetic control estimator \citep{Xu:17} in Section \ref{sec:empirical}.

\subsection{Incorporating auxiliary covariates}
\label{sec:auxiliary}

We have focused on matching pre-treatment values of the outcome variable. In practice, we typically observe a set of auxiliary covariates as well. 
For example, in the study of Proposition 99, \cite{ADH:10} consider the following covariates:
average retail price of cigarettes, per capita state personal income, per capita beer consumption, and the percentage of the population age 15–24.

It is natural to incorporate auxiliary covariates in applying the SRC method. 
For unit $j$, denote $\mbf{x}_j$ as a $(p\times 1)$ vector of observed covariates that are not affected by the intervention. 
Let $\mbf{X}=(\mbf{x}_1,\mbf{x}_2,\cdots,\mbf{x}_{J+1})$. 
Analogous to the SC method (\citealp{ADH:10}), 
We define the augmented $(T_d\times 1)$, where $T_d=T_0+p$, vector of pre-intervention characteristics for the treated unit
$\mbf{z}_1=(\mbf{y}_1^{\top},\mbf{x}_1)^{\top}\in\mathbb{R}^{T_d}$. 
Similarly, $\mbf{Z}_0$ is a $(T_d\times J_0)$ matrix that contains the same variables for the control units. 
Because the auxiliary covariates are included in the $T_d$ predictors,  
a positive definite and diagonal matrix $\mbf{V} \in \mathbb{R}^{T_d\times T_d}$ is required to reflect relative importance of each predictor.  
A common way for selecting $\mbf{V}$ is to minimize the mean squared prediction error of the outcome variable for the pre-intervention periods (\citealp{AG:03}, and \citealp{ADH:10}). 
Once $\mbf{V}$ is obtained, we denote $\tilde{\mbf{z}}_1=\mbf{V}^{1/2}\mbf{z}_1$ and $\tilde{\mbf{Z}}_0=\mbf{V}^{1/2}\mbf{Z}_0$.

We apply Algorithm \ref{alg:main} on $\tilde{\mbf{z}}_1$ and $\tilde{\mbf{Z}}_0$ to obtain $\hat{w}_j^{(\mbf{z})}$ and $\hat{\theta}_j^{(\mbf{z})}$, and then obtain the SRC estimator 
\begin{equation}\label{eq:sSRC-estimator-z}
  \hat{Y}_{1t}(0)=T_0^{-1}\sum\limits_{t=1}^{T_0}Y_{1t}+\sum_{j=2}^{J+1}\hat{w}_j^{(\mbf{z})}\hat{\theta}_j^{(\mbf{z})}(Y_{jt}-T_0^{-1}\sum\limits_{t=1}^{T_0}Y_{jt}), \ \ t\in \{1,\cdots, T\}.
\end{equation}
We summarize it as Algorithm \ref{alg:Z} below.
 \begin{algorithm}
 \begin{algorithmic}
  \item Step 1: Combine $\mbf{y}_1$ with $\mbf{x}_1$ to obtain $\mbf{z}_1$, and similarly combine $\mbf{Y}_0$ with $\mbf{X}_0$ to obtain $\mbf{Z}_0$. 
    \item Step 2: Obtain $\mbf{V}$ and denote $\tilde{\mbf{z}}_1=\mbf{V}^{1/2}\mbf{z}_1$ and $\tilde{\mbf{Z}}_0=\mbf{V}^{1/2}\mbf{Z}_0$. 
 \item  Step 3: Perform Algorithm \ref{alg:main} on $\tilde{\mbf{z}}_1$ and $\tilde{\mbf{Z}}_0$, and then obtain $\hat{w}_j^{(\mbf{z})}$ and $\hat{\theta}_j^{(\mbf{z})}$. 
 \item  Step 4: Obtain the SRC estimator according to \eqref{eq:sSRC-estimator-z}. 
 \end{algorithmic}
 \caption{The SRC estimator when auxiliary covariates are incorporated}
 \label{alg:Z}
 \end{algorithm}

\section{Empirical Studies}
\label{sec:empirical}

In this section, 
we conduct extensive Monte Carlo simulation studies to assess the performance of various methods, finding where and how the SRC estimator performs compared to existing estimators, and subsequently we perform an empirical analysis on a real dataset to examine the behavior of the SRC method.

\subsection{Simulation Studies}
Now we investigate the finite sample performance of alternative estimators in the simulation experiments using a factor model. 
We compare several representative synthetic estimators, including: 
(a) the original SC (SC) in \cite{ADH:10}, (b) the de-meaned SC (dSC) in \cite{FP:21}, (c) the augmented SC (ASC) in \cite{Ben:21}, (d) the generalized synthetic control (GSC) in \cite{Xu:17}, (e) the matching and SC (MASC) in \cite{Kellogg:21}, (f) OLS in \cite{Hsiao:12},  and (h) the constrained lasso (lasso) in \cite{CWZ:21}.

In this experiment, all units are generated according to the factor model as follows
$$Y_{jt}(0)=\alpha_t+\lambda_tf_j+\epsilon_{jt},$$
where unobserved factors $\lambda_t\sim \mathcal{N}(0,1)$. 
To assess the impact of fixed time effects $\alpha_t$ and the influence of the gap between $f_1f_j$ and $f_j^2$ ($j\neq 1$), 
we consider three model settings: 
(1) \emph{F1}: $\alpha_t=0$ for all $t$; $f_j=1$ for $j=1,\cdots, 7$ and $f_j=0$ for $j=8,\cdots, J+1$. 
(2) \emph{F2}: $\alpha_t=0$ for all $t$; $f_1=3$ and $f_j=1$ for $j=2,\cdots, J+1$. 
(3) \emph{F3}: $\alpha_t\sim \mathcal{N}(0,1)$ for all $t$; $f_1=3$ and $f_j=1$ for $j=2,\cdots, J+1$. 
For \emph{F1}, there are no fixed time effects; and all nonzero factor loadings are set to be
ones, so both treated and control units with nonzero loadings are drawn from a common distribution. 
In \emph{F2}, there are no fixed time effects as \emph{F1}; however, the treated and control units are drawn from heterogeneous distributions: the treated unit has a loading of 3, while control units have loadings equal to 1.  
for \emph{F3}, the factors loadings are the same as \emph{F2}, but fixed time effects $\alpha_t\sim \mathcal{N}(0,1)$ are added. 
We set $T=50$ with $T_0=40$ and $J=20$. 
The errors $\epsilon_{jt}\sim \mathcal{N}(0,\sigma^2)$, where we use three values of $\sigma$ values, 1, 0.5, and 0.1, to investigate the impact of $\sigma$. 
We also consider heteroscedastic noise following an autoregressive process, specified as $\epsilon_{jt} = 0.6\epsilon_{j(t-1)} + \epsilon_{jt}^w$, where $\epsilon_{jt}^w$ denotes independent white noise drawn from $\mathcal{N}(0, \sigma^2)$. The corresponding results, reported in Figure~\ref{simulation-ARnoise} of the Appendix, are consistent with those obtained under the homogeneous noise setting.

To evaluate each estimator, we compute the mean squared prediction error (MSPE), which is defined as 
$\text{MSPE}=(T-T_0)^{-1}\sum\nolimits_{t=T_0+1}^T\|\hat{Y}_{1,t}(0)-Y_{1,t}(0)\|^2$, by calculating the average loss across 500 simulations. The results are reported in Table \ref{simulation1}. 
For the homoscedastic setting \emph{F1}, all methods perform similarly.
In the heteroscedastic setting \emph{F2}, SC, dSC, ASC, GSC, MASC, and Lasso perform poorly, whereas SRC and OLS perform well, with SRC outperforming OLS. 
In the heteroscedastic setting, the factor loading of the treated unit lies outside the convex hull of the factor loadings for the control units, so no set of weights on the simplex can balance the factor loadings. While synthetic control estimators are expected to struggle in such settings, OLS tends to perform well because it permits unrestricted extrapolation. 
In contrast, our SRC approach adjusts for the imbalance in factor loadings between the treated and control units through unit regression, while also controlling extrapolation error by preserving the key structure of synthetic controls. As a result, SRC is expected to outperform OLS, and the results in Table \ref{simulation1} confirm this expectation.
When fixed time effects are present in the heteroscedastic setting \emph{F3}, OLS performs poorly, as it fails to account for the fixed time effects. In contrast, SRC achieves the best performance.  
These observations regarding SRC are consistent with the findings in Section \ref{sec:rational}, suggesting that its improvement is expected to increase as the disparity between $f_1f_j$ and $f_j^2$ ($j\neq 1$), becomes more pronounced, while still preserving the advantages of the sparse synthesizing approach characteristic of SC-type methods. 

\begin{table}[!ht]
\caption{Post-intervention MSPE of alternative estimators.}
\begin{center}
\begin{tabular}{c | c c c c c c c c c}
\hline
$(\text{model},\sigma)$ & \scriptsize SC & \scriptsize dSC & \scriptsize ASC & \scriptsize GSC  & \scriptsize MASC & \scriptsize OLS & \scriptsize lasso & \scriptsize SRC\\
\hline
$(\emph{F1},1)$ &1.426 &1.282&1.355&\bf1.221&1.345&2.521&1.320&1.446\\
$(\emph{F1},0.5)$ &0.571&\bf0.303&0.384&0.303&0.372&0.534&0.311&0.348\\
$(\emph{F1},0.1)$ & 0.223&0.014&0.016&\bf0.013&0.017&0.026&0.014&0.017\\
\hline
$(\emph{F2},1)$ &5.546&5.509&5.543&5.786&5.716&2.811&5.539&\bf1.932\\
$(\emph{F2},0.5)$ &4.731&4.708&4.757&4.738&4.031&0.656&4.664&\bf0.453\\
$(\emph{F2},0.1)$ & 4.403&4.293&4.397&4.421&4.811&0.033&4.295&\bf0.021\\
\hline
$(\emph{F3},1)$ &3.629&3.591&3.839&3.616&3.901&5.185&3.751&\bf2.682\\
$(\emph{F3},0.5)$ & 3.460&3.563&3.744&4.148&3.598&5.159&3.517&\bf2.319\\
$(\emph{F3},0.1)$ & 3.317&3.270&3.417&3.283&3.374&3.636&3.281&\bf1.747\\
\hline
\end{tabular}
\end{center}
\label{simulation1}
\end{table}

\begin{table}[!ht]
\caption{Post-intervention MSPE of alternative estimators when the number of units is large.}
\begin{center}
\begin{tabular}{ c | c c c c c c c c c}
\hline
 $(\text{model},\sigma)$ & \scriptsize SC & \scriptsize dSC & \scriptsize ASC & \scriptsize GSC &  \scriptsize MASC & \scriptsize OLS & \scriptsize lasso & \scriptsize SMC\\
\hline
$(\emph{F1},1)$&1.377&1.453&\bf1.373&1.400&2.286& 7.570&1.494&1.454\\
$(\emph{F1},0.5)$&0.327&0.341&0.347&0.314&0.261& 1.211&0.357&\bf0.314\\
$(\emph{F1},0.1)$&0.012&0.012&0.011&\bf0.010&0.011& 0.054&0.011&0.012\\
\hline
$(\emph{F2},1)$&5.059&5.116&5.168&5.398&5.254& 5.155&5.047&\bf1.813\\
$(\emph{F2},0.5)$&4.963&5.052&5.068&4.939&2.300& 1.439&5.104&\bf0.745\\
$(\emph{F2},0.1)$&3.501&3.512&3.723&3.859&1.936& 0.062&3.481&\bf0.016\\
\hline
$(\emph{F3},1)$&6.226&6.583&7.837&7.168&6.502&16.818&6.635&\bf4.221\\
$(\emph{F3},0.5)$&5.008&5.138&5.102&5.241&8.844&14.849&5.139&\bf3.032\\
$(\emph{F3},0.1)$&3.036&3.469&3.374&3.926&3.126&10.476&3.423&\bf2.048\\
\hline
\end{tabular}
\end{center}
\label{simulation2}
\end{table}

\textbf{Noise-to-Signal Impact.}
Comparing the results across various $\sigma$ values, we find that the above observations hold true. 
Notably, we observe that the MSPE values of both the SRC and OLS estimators approach zero as $\sigma$ decreases from 1 to 0.1 when no fixed time effects are present. In contrast, this convergence is not observed for the other methods under heteroscedastic settings. The superior performance of OLS in the case of \emph{F2} is consistent with the findings reported in \citet{Hsiao:12}.

\textbf{Effectiveness of Unit Screening.}
To assess the performance of unit screening when the number of units is large, we consider the case $T = 50, J = 50$, where $T_0 <J$. For this setting, we apply the SRC estimator using Algorithm \ref{alg:HD} with SIRS preprocessing on the screening units, as described in Section \ref{sec:HD-setting}. The results, reported in Table \ref{simulation2}, show that SRC performs well for \emph{F1} and achieves superior performance for \emph{F2} and \emph{F3}, highlighting the effectiveness of the SIRS preprocessing step within the extended SRC framework.

\textbf{Nonlinear Factor Model Case.}
We also assess the performance of a nonlinear factor model, where 
all units are generated according to the factor model 
$Y_{1t}(0)=\alpha_t+\lambda_tf_1+\epsilon_{1t}$ and $Y_{jt}(0)=\alpha_t^2+\lambda_t^2f_j+\epsilon_{jt}$ for $j=2,\cdots, J+1$. 
The factors and loadings are set as in \emph{F3}.  
The results are reported in Table \ref{simulation-nonlinear}, 
showing that the observation is consistent with Example 2 in Section \ref{sec:rational}. 
The SC, dSC, ASC, GSC, and MASC methods performs poorly, and 
OLS also works poorly due to the lack of restrictions on the coefficients of the predictors. 
The lasso method performs better than these methods as it can identify unrelated predictors. 
Among all methods, our approach works best, demonstrating that 
the unit regression step enhances the predictive power of synthetic controls. 

\begin{table}[ht]
\caption{Post-intervention MSPE of alternative estimators under the nonlinear factor model.}
\begin{center}
\begin{tabular}{c | c c c c c c c c}
\hline
$\sigma$ & \scriptsize SC & \scriptsize dSC & \scriptsize ASC & \scriptsize GSC  & \scriptsize MASC & \scriptsize OLS & \scriptsize lasso & \scriptsize SRC\\
\hline
$1$ &16.98&13.95&17.77&14.19&21.70&23.11&11.22&\bf11.16\\
$0.5$ &20.31&16.43&21.37&16.34&16.50&21.60&\bf11.46&11.66\\
$0.1$ &15.07&12.35&15.58&12.76&12.38&19.02&10.05& \bf9.830\\
\hline
\end{tabular}
\end{center}
\label{simulation-nonlinear}
\end{table}

\textbf{Case of a Single-Unit Synthetic Control.}
Finally, we consider two cases beyond the standard setting: the treated unit lies outside the convex hull of the donor pool, and the synthetic control coincides with a single donor unit, by using the model $Y_{jt}(0)=\lambda_tf_j+\epsilon_{jt}$.  
Here $\lambda_t\sim \mathcal{N}(0,1)$, and the errors $\epsilon_{jt}\sim \mathcal{N}(0,\sigma^2)$. 
We examine a set of factor loadings $f_j$: $f_1=2$, $f_2=2$, $f_j=1$ for $j=3,\cdots, J+1$; 
We set $T=50$ with $T_0=40$ and $J=20$. 
The case corresponds to the case where the synthetic control coincides with a single donor unit.  The results are reported in Table \ref{simulation-hull}. SRC outperforms the other methods in this case.

\begin{table}[!ht]
\caption{An experiment under the case of a synthetic control with a single donor}
\begin{center}
\begin{tabular}{c | c c c c c c c c}
\hline
$(\sigma)$ & \scriptsize SC & \scriptsize dSC & \scriptsize ASC  &  \scriptsize GSC & \scriptsize MASC & \scriptsize OLS & \scriptsize lasso & \scriptsize SRC\\  
\hline
$(1)$&1.742&1.723&1.774&1.817&2.270&2.376&2.134&\bf1.655\\
$(0.5)$&0.394&0.420&0.400&0.458&0.536&0.560&0.786&\bf0.352\\
$(0.1)$&0.020&0.021&0.020&0.021&0.022&0.026&0.529&\bf0.015\\
\hline
\end{tabular}
\end{center}
\label{simulation-hull}
\end{table}

\subsection{The Basque dataset}
We study the effect of terrorism on per capita GDP in Basque, Spain.  
The Basque dataset is from \cite{AG:03}. It consists of per capita GDP of 17 regions in Spain from 1955 to 1997, and 12 other covariates of each region over the same time interval, representing education, investment, sectional shares, and population density in each region.
We incorporate auxiliary covariates which include averages for the 13 characteristics from 1960 to 1969, and scale each covariate so that it has equal variance of outcomes. 
In this study, the treated unit is the Basque Country, and the treatment is the onset of separatist terrorism, which begins in 1970.

\textbf{Placebo Analysis.}
Similar to \cite{AG:03}, we conduct a placebo study to compare alternative estimators in the real data. 
We perform placebo analyses on each region, excluding Basque, as the placebo region. We calculate the mean squared prediction error (MSPE) for each region by taking the differences between its actual and fitted outcome paths in each of the post-period years (1970-1997), squaring these differences and then averaging them among these years. The results of our analysis are presented in Table \ref{realdata-placebo}, which shows that, on average, SRC tends to have the lowest MSPE. 
In addition, we include the pre-period fit of these estimator in Table \ref{realdata-fit} of Appendix B to further demonstrate their performance. 
Interestingly, we observe that SRC does not exhibit the best pre-period fit on average (it is the second best), while GSC demonstrates the best pre-period fit on average. 
This observation suggests that SRC is less prone to over-fitting compared to ASC.

\begin{table}[!ht]
\caption{Performance (MSPE) of alternative estimators in the placebo study.}
\begin{center}
\begin{tabular}{c | c c c c c c c c}
\hline
region & \scriptsize SC & \scriptsize dSC & \scriptsize ASC  & \scriptsize GSC & \scriptsize MASC  & \scriptsize OLS & \scriptsize lasso &  \scriptsize SRC\\
\hline
Andalucia & 0.41 & 0.15 & 0.17  & \bf0.02 & 0.32 &  5.18 &  0.13 & 0.16\\
Aragon & 0.03 & 0.12 & 0.06  & 0.06 & 0.04 &  0.03 & 0.02 &  \bf0.01\\
Asturias & 0.71 & 0.56 & 3.40 & \bf0.06 & 0.73 & 0.30 & 0.44 &  0.77\\
Baleares & 2.12 & 3.68 & 1.24 & 0.57 & 2.12 & 2.51 & 4.73 &  \bf0.56\\
Canarias & 0.07 & 0.10 & 0.35 & 0.45 & 0.07 & 0.96 & \bf0.02 &  0.29\\
Cantabria & 0.37 & 0.65 & 0.90 & 0.81 & 0.34 & 1.87 &  0.56 &  \bf0.11\\
Leon & \bf0.01 & 0.12 & 0.08 & 0.06 & 0.01 & 0.15 &  0.13 &  0.05\\
Mancha & 0.07 & \bf0.02 & 0.04 & 0.74 & 0.04 &  0.49 &  0.39 &  0.34\\
Cataluna & 0.44 & \bf0.03 & 0.14 & 0.08 & 0.44 & 0.73 &  1.33 & 0.25\\ 
Valenciana & 0.15 & 0.14 & 0.09 & 1.79 & 0.04 & 0.08 &  \bf0.03 &  0.29\\
Extremadura & 0.74 & \bf0.06 & 0.17 & 0.21 & 0.74 & 0.99 & 0.63 &  0.08\\
Galicia & \bf0.01 & 0.02 & 0.04 & 0.19 & \bf0.01 & 0.07 & 0.01 &  0.04\\
Madrid & \bf0.11 & 0.38 & 3.75 & 3.22 & \bf0.11 & 0.48 &  0.16 & 0.28\\ 
Murcia & 0.21 & 0.27 & 0.07 & 0.18 & 0.20 &  1.18 &  \bf0.03 &  0.19 \\
Navarra & 0.04 & 0.04 & 0.03 & 0.06 & 0.05 &  0.04 &  0.12 &  \bf0.03\\
Rioja & 0.04 & \bf0.03 & 0.08 & 0.04 & 0.19 & 0.09 &  0.30 & 0.04\\
\hline
average & 0.35 & 0.40 & 0.66 & 0.53  & 0.34 & 0.95 & 0.56 & \bf0.22\\
\hline
\end{tabular}
\end{center}
\label{realdata-placebo}
\end{table}


\begin{figure}[!ht]
  \centering
  \includegraphics[width=4.5in]{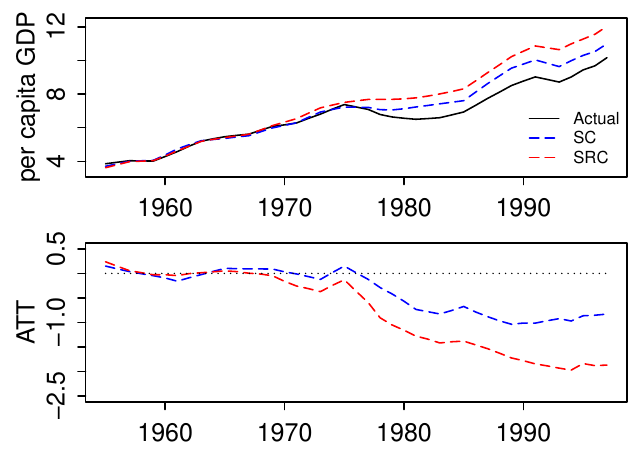}
  \caption{Study of the Basque Country. Upper: Actual and counterfactual per capita GDP of the Basque Country. 
  Bottom: The difference of actual and counterfactual values, representing the average treatment effect on the treated (ATT).}
  \label{fig:synthetic}
\end{figure}

\begin{table}[!ht]
\caption{Estimates of weights for alternative estimators in the Basque study.
The comprehensive weights $\theta_jw_j$ and the coefficients are reported for SRC and OLS, respectively.}
\begin{center}
\begin{tabular}{c | c c c c c c c}
\hline
region & \scriptsize SC & \scriptsize dSC & \scriptsize MASC & \scriptsize OLS & \scriptsize lasso &  \scriptsize SRC\\ 
\hline
Andalucia & 0 & 0 & 0 & 0.217 & -0.112 & 0\\
Aragon & 0 & 0 & 0 & -3.059 & 0 &  0\\
Asturias & 0 & 0 & 0 & 1.460 & 0 &  0.001\\
Baleares & 0 & 0.582 & 0 & -0.365 & 0.365 &  0\\
Canarias & 0 & 0 & 0 & -0.245 & 0 &  0\\ 
Cantabria & 0 & 0 & 0 & 0.048 & 0.048 & 0.311\\
Leon & 0 & 0 & 0 & 0.080 & 0 & 0\\
Mancha & 0 & 0 & 0 & 0.861 & 0 & 0\\
Cataluna & 0.851 & 0 & 0.851 & -0.174 &  0.174 &  0.028\\
Valenciana & 0 & 0 & 0 & 1.766 & 0.065 &  0\\ 
Extremadura & 0 & 0 & 0 & -0.401 & 0 &  0\\
Galicia & 0 & 0 & 0 & -0.396 & 0 & 0\\
Madrid & 0.149 & 0.418 & 0.149 & 0.533 &  0 &  0.276\\ 
Murcia & 0 & 0 & 0 & -0.547 & 0 & 0\\
Navarra & 0 & 0 & 0 & 1.302 & 0.234 & 0\\
Rioja & 0 & 0 & 0 & 0.673 & 0 &  0.587\\ 
\hline
Intercept & - & -0.335 & - & -2.580 & 1.097 &  0.047\\ 
\hline
\end{tabular}
\end{center}
\label{basque-syn}
\end{table}

\textbf{Synthetic Basque.}
We estimate the effect of exposure to terrorism on GDP per capita in Basque, Spain. 
We present the GDP per capita for both Basque and its synthetic control, generated using the SC and SRC methods, in Figure \ref{fig:synthetic}. We also show the difference between the actual values and the synthetic values. 
Due to saving space and for ease of comparison, we just report the comparison between SRC and the original SC method, ignoring other estimators.  
To gain a better understanding of the SRC method, we inspect the comprehensive weight $\theta_jw_j$ assigned to each control unit $j$ and compare it with the weights obtained from other methods. The results are reported in Table \ref{basque-syn}. The SC, dSC, and MASC methods assign non-zero weights to just two control units. The lasso method assigns non-zero weights to six control units. Comparing to them, SRC assigns non-zero weights to five units. 
It is important to note that OLS does not have zero weight since it is not subject to constraints. These findings suggest that, in contrast to other constrained methods, the greater flexibility of SRC may contribute to enhanced predictive power.

\section{Discussion}
\label{sec:discussion}
This paper makes three contributions: (1) We propose a simple and effective method, Synthetic Regressing Control, by synthesizing the regressed controls. 
(2) We determine the weights by minimizing the unbiased risk estimate criterion. 
We demonstrate that this method is asymptotically optimal, achieving the minimum loss of the infeasible best possible synthetic estimator.
(3) We expand the method to cases where control units is more than time periods, and incorporating auxiliary covariates.

There are several potential directions for future work. 
First, our working model framework may not hold for non-stationary data, as discussed in Section \ref{sec:rational}, because the unit-sum constraint is removed. When the unit-sum constraint is imposed on the weights, the $C_p$ criterion for searching $\w$ in Section \ref{sec:weight} degenerates to minimizing the pre-intervention loss.  A more efficient approach based on an alternative criterion is therefore needed. 
Second, we focus on the simple linear regression to reduce the issue of imperfect fit. Consequently, the synthesis method relies on the simple linear regression. However, the SRC method  may be applicable for more complex data structures, such as discrete, count, or hierarchical outcomes, and nonlinear relationships. Therefore, extending the method to broader regression models is an interesting direction. 
Third, for settings with multiple treated units, we can fit SRC separately for each treated unit, as in \cite{Abadie:21}. However, this approach brings a loss of efficiency due to the correlation of treated units. Therefore, efficiently extending the method to multiple treated units is a worthy direction. 
Fourth, if a set of auxiliary covariates is available, we pool the auxiliary covariates and the outcomes together to conduct the SRC method. However, this approach may bring extra risk when the linear approximation relation in the covariates is different from that in the outcomes. Therefore, 
exploring ways to incorporate auxiliary covariates into the SRC method while minimizing such risks is another worthy problem for future research.
Fifth, in this paper we focus on estimation and have not yet rigorously conducted inference. While the placebo permutation test \citep{ADH:10} can be applied to the SRC method, it does not account for the uncertainty arising from both the unit regression step and the estimation of the variance of the idiosyncratic shocks. Therefore, it is important to develop statistical inference for the SRC method analytically, explicitly incorporating the uncertainty from the unit regression step. 
Finally, extending it to more complicated situations, such as staggered adoption where units take up the treatment at different times (\citealp{Ben:22}), is another challenging direction.

\section*{Competing interests}
No competing interest is declared.

\section*{Acknowledgements}

I sincerely thank the Associate Editor, Dr. Youjin Lee, and the two anonymous reviewers for their insightful and constructive comments, which substantially improved the manuscript.  I also thank Kaspar W\"uthrich for valuable discussions on the preliminary version. 

\vskip 0.2in
\bibliography{SCM}

\newpage
\appendix
\renewcommand{\theequation}{A.\arabic{equation}}
\renewcommand{\thesection}{A}
\setcounter{equation}{0}


\section*{Appendix A: Proofs of Results}
\label{sec:proof2}

\subsection{Asymptotic optimality for in-sample predictions}

\begin{theorem}
\label{th:opt}
Under Model \eqref{model-po}, 
assume that (1) $\max_t\E{\epsilon_t^4}\leq c_1<\infty$ for some constant $c_1$,
(2) $\|\bmu_1\|^2/T_0\leq c_2<\infty$ for some constant $c_2$, 
and (3) $J^{-1}\|\bmu_1-\mbf{Y}_0[\text{diag}(\mbf{Y}_0^{\top}\mbf{Y}_0)]^{-1}\mbf{Y}_0^{\top}\bmu_1\|^2\rightarrow_p \infty$ as $T_0\rightarrow \infty$, 
then as $T_0\to \infty$, 
\begin{align*}
\frac{L(\hat{\w})}{\inf_{\w\in \mathcal{H}_J}L(\w)}\rightarrow_p 1.
\end{align*}
\end{theorem}

\begin{proof} 
For simplifying the notation, we assume, without loss of generality, that $\mbf{y}_j, \forall j=2,\cdots,J+1$ and $\mbf{\mu}_1$ are centered, i.e., 
$\mbf{1}^{\top}\mbf{y}_j=0$ and $\mbf{1}^{\top}\mbf{\mu}_1=0$. 

The solution $\hat{\sbf{\theta}}=(\hat{\theta}_2,\cdots,\hat{\theta}_{J+1})^{\top}$ satisfies:
\begin{align*}
\hat{\sbf{\theta}} 
=& \left[\text{diag}(\mbf{Y}_0^{\top}\mbf{Y}_0)\right]^{-1}\mbf{Y}_0^{\top}\mbf{y}_1. 
\end{align*}
Inserting $\hat{\sbf{\theta}}$ into the criterion $\mathcal{C}_c(\w)$, we have 
\begin{align}\label{Eqn:Cp}
\mathcal{C}_c(\w) &= \left\|\mbf{y}_1-\mbf{Y}_0\mbf{W}\left[\text{diag}(\mbf{Y}_0^{\top}\mbf{Y}_0)\right]^{-1}\mbf{Y}_0^{\top}\mbf{y}_1\right\|^2+2\hat{\sigma}^2\w^{\top}\mbf{1}. 
\end{align}

By the above analysis on Algorithm \ref{alg:main}, we just need to prove the asymptotic optimality by analyzing \eqref{Eqn:Cp}. 
In the theorem, we have the following conditions: 
\begin{align*}
&(C1) \quad \max\nolimits_t\E{\epsilon_t^4}\leq c_1<\infty \text{ for some constant } c_1.\notag\\
&(C2) \quad \|\bmu_1\|^2/T_0\leq c_2<\infty \text{ for some constant } c_2.\notag\\
&(C3) \quad J^{-1}\|\bmu_1-\mbf{Y}_0[\text{diag}(\mbf{Y}_0^{\top}\mbf{Y}_0)]^{-1}\mbf{Y}_0^{\top}\bmu_1\|^2\rightarrow \infty 
\text{ as }T_0\rightarrow \infty. 
\end{align*}
Let 
$\mbf{M}(\w)=\mbf{Y}_0\mbf{W}\left[\text{diag}(\mbf{Y}_0^{\top}\mbf{Y}_0)\right]^{-1}\mbf{Y}_0^{\top}$ and $\A(\w)=\I-\mbf{M}(\w)$.  
We have 
\begin{align}
 &\hat\y_1(\w)=\mbf{M}(\w)\y_1, \notag\\
 &L(\w)=\|\hat\y_1(\w)-\bmu_1\|^2=\|\mbf{M}(\w)\y_1-\bmu_1\|^2. 
 \label{eqn:L}
\end{align}
From \eqref{eqn:L}, a simple calculation follows 
\be
\label{Rw}
R(\w)=E\{L(\w)  | \mbf{Y}_0 \}= \left\|\A(\w)\bmu_1\right\|^2 + \tr\left\{\mbf{M}(\w)\bSig\mbf{M}^{\top}(\w)\right\}\geq \left\|\A(\w)\bmu_1\right\|^2.
\ee
Thus, combing \eqref{Rw} and Condition C3, we have 
\begin{align}\label{infR}
\inf_{w\in\mathcal{H}_J}J^{-1}R(\w)\rightarrow \infty. 
\end{align}

Denote $\bSig=\sigma^2\mbf{I}_{T_0}$. 
From \eqref{Eqn:Cp} and \eqref{eqn:L}, we have
\begin{align*}
\mathcal{C}_c(\w)-L(\w)
&=\|\hat\y_1(\w)-\bmu_1-\e\|^2+2\hat{\sigma}^2 \w^{\top}\mbf{1}-\left\|\mbf{M}(\w)\y_1-\bmu_1\right\|^2\notag\\
&=-2\left\{\hat\y_1(\w)-\bmu_1\right\}^\top\e+2\hat{\sigma}^2 \w^{\top}\mbf{1}+\left\|\e\right\|^2\notag\\
&=-2\e^\top\mbf{M}(\w)\y_1+2\hat{\sigma}^2 \w^{\top}\mbf{1}+2\bmu_1^\top\e+\left\|\e\right\|^2\notag\\
&=-2\e^\top\mbf{M}(\w)\bmu_1-2\e^\top\mbf{M}(\w)\e+2\hat{\sigma}^2 \w^{\top}\mbf{1}+2\bmu_1^\top\e+\left\|\e\right\|^2, 
\end{align*}
where the last three terms in the last line do not involve $\w$. 
Following the proof approach in \citet{Li:87} and \citet{Zhu:23}, if we can show that $\e^\top\mbf{M}(\w)\bmu_1$ and $\hat{\sigma}^2 \w^{\top}\mbf{1}-\e^\top\mbf{M}(\w)\e$ are negligible (compared to $L(\w)$) uniformly for 
any $\w\in\mathcal{H}_J$, then the asymptotic optimality for $\hat{\w}$ is established. 
More precisely, it remains to show that in probability, 
$$\supw
\left|\frac{\e^\top\mbf{M}(\w)\bmu_1}{R(\w)}\right| \rightarrow 0,$$
$$\supw
\left|\frac{\hat{\sigma}^2 \w^{\top}\mbf{1}-\e^\top\mbf{M}(\w)\e}{R(\w)}\right| \rightarrow 0,$$
and 
$$\supw
\left|\frac{L(\w)}{R(\w)}-1\right| \rightarrow 0.$$

From \eqref{eqn:L} and \eqref{Rw}, we have
\begin{align*}
R(\w)-L(\w)
=&\left\|\A(\w)\bmu_1\right\|^2 +  \tr\left\{\mbf{M}(\w)\bSig\mbf{M}^{\top}(\w)\right\}-
\left\|\mbf{M}(\w)\y_1-\bmu_1\right\|^2\notag\\
=&\left\|\A(\w)\bmu_1\right\|^2 +  \tr\left\{\mbf{M}(\w)\bSig\mbf{M}^{\top}(\w)\right\}-
\left\|\A(\w)\bmu_1-\mbf{M}(\w)\e\right\|^2\notag\\
=&-\left\|\mbf{M}(\w)\e\right\|^2 + \tr\left\{\mbf{M}(\w)\bSig\mbf{M}^{\top}(\w)\right\}+
2\e^\top\mbf{M}(\w)\A(\w)\bmu_1.
\end{align*}
Thus, to prove the asymptotic optimality,
We need only to verify that
$$
\begin{array}{lll}
&(\text{\romannumeral1})\
\supw
R^{-1}(\w)\left|
\e^\top\mbf{M}(\w)\bmu_1\right|=o_P(1), \\
&(\text{\romannumeral2})\  
\supw
R^{-1}(\w)[\hat{\sigma}^2 \w^{\top}\mbf{1}-\e^\top\mbf{M}(\w)\e]=o_P(1),\\
&(\text{\romannumeral3})\
\supw
R^{-1}(\w)\left\|\mbf{M}(\w)\e\right\|^2=o_P(1),\\
&(\text{\romannumeral4})\ \supw
R^{-1}(\w)\tr\left\{\mbf{M}(\w)\bSig\mbf{M}^{\top}(\w)\right\}=o_P(1),\\
&(\text{\romannumeral5})\ \supw
R^{-1}(\w)\left|\e^\top\mbf{M}(\w)\A(\w)\bmu_1\right|
=o_P(1).
\end{array}
$$
Let $\mbf{H}=\mbf{Y}_0\left[\text{diag}(\mbf{Y}_0^{\top}\mbf{Y}_0)\right]^{-1}\mbf{Y}_0^{\top}$. 
A direct simplification yields
\begin{align}
\mbf{M}(\w)&\le \mbf{H}, \label{H}
\end{align}
where the inequality is in Loewner ordering.
Given Condition C1, applying Theorem 2 of \cite{W:60} leads to that, for some constant $0<\eta<\infty$,
\begin{align*}
\E{\left(\e^\top\mbf{H}\e-\E{\e^\top\mbf{H}\e}\right)^{2}  | \mbf{Y}_0 }
\leq \eta(\tr\{\mbf{H}\bSig\mbf{H}\}). 
\end{align*}
It follows 
\begin{align}\label{L-1}
  \mbf{\epsilon}^{\top}\mbf{H}\mbf{\epsilon}
   = & \tr\{\bSig\mbf{H}\}
     +O_P\bigg\{\sqrt{\tr\{\bSig\mbf{H}^2\}}\bigg\}
  = O_p(J),
\end{align}
where the last step is from  
$\tr\{\bSig\mbf{H}\}=\tr\{\sum\nolimits_{j=2}^{J+1}w_j(\mbf{y}_j^{\top}\bSig\mbf{y}_j)(\mbf{y}_j^{\top}\mbf{y}_j)^{-1}\}
=O_p(J)$ and, similarly, $\tr\{\bSig\mbf{H}\}=O_p(J)$. 
From \eqref{infR} and \eqref{L-1}, we have
\be
\label{ePe}
\supw\frac{\e^\top\mbf{H}\e}{R(\w)}\leq \frac{O_P(J)}{R(\w)}=o_P(1).
\ee
From (\ref{infR}) and (\ref{ePe}), we have
\be
\supw\frac{
\left|
\e^\top\mbf{H}\mbf{A}(\w)\bmu_1\right|}{R(\w)}\leq
\supw\frac{\left\|\mbf{H}\e\right\|\left\|\A(\w)\bmu_1\right\|}{R(\w)}\leq
\supw\frac{\left\|\mbf{H}\e\right\|}{R^{1/2}(\w)}
=o_P(1).\n
\ee
Thus, (\romannumeral1) is proved.

From Condition C2, we have
\be
\label{hatsigma}
\hat\sigma^2 =\frac{\|\mbf{y}_1-\mbf{Y}_0(\mbf{Y}_0^{\top}\mbf{Y}_0)^{-1}\mbf{Y}_0^{\top}\mbf{y}_1\|^2}{T_0-J} = O_P(1).
\ee
It follows that 
\be
\supw\frac{
|\hat{\sigma}^2 \w^{\top}\mbf{1}-\e^\top\mbf{M}(\w)\e|
}{R(\w)}
\leq \supw\frac{\hat{\sigma}^2 \w^{\top}\mbf{1}}{R(\w)}+\supw\frac{
\e^\top\mbf{H}\e}{R(\w)}=o_P(1), \n
\ee
where the last step is from \eqref{infR}, \eqref{ePe}, and (\ref{hatsigma}). 
Thus, (\romannumeral2) is proved.

From (\ref{H}), we have
\be
\supw\frac{\|\mbf{M}(\w)\e\|^2}{R(\w)}\leq
\supw\frac{
\e^\top\mbf{H}\e}{R(\w)}
=o_P(1),\n
\ee
where the last step is from \eqref{infR}. 
Thus, (\romannumeral3) is proved.

From (\ref{H}), we have
\be\supw\frac{
\tr\left\{\mbf{M}(\w)\bSig\mbf{M}^{\top}(\w)\right\}}{R(\w)} \leq \supw\frac{\tr\{\mbf{H}\bSig\}}{R(\w)}=o_P(1), \n
\ee
where the last step is from \eqref{infR} and Condition C1. 
Thus, (\romannumeral4) is proved.

From (\ref{Rw}) and (\ref{H}), we have
\be
 \supw\frac{\left|\e^\top\mbf{M}(\w)\A(\w)\bmu_1\right|}{R(\w)}
\leq
\supw\frac{\left\|\mbf{M}(\w)\e\right\|\left\|\A(\w)\bmu_1\right\|}{R(\w)}
\leq
\supw\frac{\left\|\mbf{H}\e\right\|}{R^{1/2}(\w)}
=o_P(1), \n
\ee
where the last step is from \eqref{ePe}. It means that (\romannumeral6) is proved.
Therefore, the theorem is proved.

\end{proof}

\subsection{Proof of Theorem \ref{th:opt-out}}

Denote 
$\mbf{M}^{(o)}(\w)=\mbf{Y}_0^{(o)}\mbf{W}\left[\text{diag}(\mbf{Y}_0^{\top}\mbf{Y}_0)\right]^{-1}\mbf{Y}_0^{\top}$ and $\A^{(o)}(\w)=\I-\mbf{M}^{(o)}(\w)$.  
We have 
\begin{align}
 &\hat\y_1^{(o)}(\w)=\mbf{M}^{(o)}(\w)\y_1, \notag\\ 
 &L^{(o)}(\w)=\|\hat\y_1^{(o)}(\w)-\bmu_1^{(o)}\|^2=\|\mbf{M}^{(o)}(\w)\y_1-\bmu_1^{(o)}\|^2. \label{eqn:L-out}
\end{align}
In the theorem, we have the following conditions: 
\begin{align}
&(C4) \quad \|\bmu_1^{(o)}\|^2/(T-T_0)\leq c_3<\infty \text{ for some constant } c_3.\notag\\
&(C5) \quad J^{-1}(T-T_0)^{-1}T_0\|\bmu_1^{(o)}-\mbf{Y}_0^{(o)}[\text{diag}(\mbf{Y}_0^{\top}\mbf{Y}_0)]^{-1}\mbf{Y}_0^{\top}\bmu_1\|^2\rightarrow_p \infty 
\text{ as }T_0\rightarrow \infty. \notag\\
&(C6) \quad \text{The model \eqref{model:mu1-model} holds.} \notag
\end{align}
Similar to \eqref{Rw}, we have 
\be
\label{Rw-out}
R^{(o)}(\w)=E\{L^{(o)}(\w)  | \mbf{Y}_0 \}
\geq \left\|\A^{(o)}(\w)\bmu_1^{(o)}\right\|^2.
\ee
Thus, combing \eqref{Rw-out} and Condition C5, we have 
\begin{align}\label{infR-out}
\inf_{w\in\mathcal{H}_J}J^{-1}(T-T_0)^{-1}T_0R^{(o)}(\w)\rightarrow \infty. 
\end{align}

Define $\Delta=\left\|\mbf{M}(\w)\y_1-\bmu_1\right\|^2 - T_0(T-T_0)^{-1}\left\|\mbf{M}^{(o)}(\w)\y_1-\bmu_1^{(o)}\right\|^2$. 
From \eqref{Eqn:Cp} and \eqref{eqn:L-out}, we have
\begin{align*}
\mathcal{C}(\w)-L^{(o)}(\w)
&=\|\hat\y_1(\w)-\bmu_1-\e\|^2+2\hat{\sigma}^2 \w^{\top}\mbf{1} - T_0(T-T_0)^{-1}\left\|\mbf{M}^{(o)}(\w)\y_1-\bmu_1^{(o)}\right\|^2\notag\\
&=-2\left\{\hat\y_1(\w)-\bmu_1\right\}^\top\e+2\hat{\sigma}^2 \w^{\top}\mbf{1}+\left\|\e\right\|^2 + \Delta \notag\\
&=-2\e^\top\mbf{M}(\w)\y_1+2\hat{\sigma}^2 \w^{\top}\mbf{1}+2\bmu_1^\top\e+\left\|\e\right\|^2 + \Delta \notag\\
&=-2\e^\top\mbf{M}(\w)\bmu_1-2\e^\top\mbf{M}(\w)\e+2\hat{\sigma}^2 \w^{\top}\mbf{1}+2\bmu_1^\top\e+\left\|\e\right\|^2 +\Delta. 
\end{align*}
Similar to the proof of Theorem \ref{th:opt}, under Conditions C1 and C4, as well as \eqref{infR-out}, we have
$$
\begin{array}{lll}
&(\text{\romannumeral1})\
\supw
T_0^{-1}(T-T_0)\left|
\e^\top\mbf{M}(\w)\bmu_1\right|/R^{(o)}(\w)=o_P(1), \\
&(\text{\romannumeral2})\  
\supw
T_0^{-1}(T-T_0)[\hat{\sigma}^2 \w^{\top}\mbf{1}-\e^\top\mbf{M}(\w)\e]/R^{(o)}(\w)=o_P(1),\\
&(\text{\romannumeral3})\
\supw
T_0^{-1}(T-T_0)\left\|\mbf{M}(\w)\e\right\|^2/R^{(o)}(\w)=o_P(1),\\
&(\text{\romannumeral4})\ \supw
T_0^{-1}(T-T_0)\tr\left\{\mbf{M}(\w)\bSig\mbf{M}^{\top}(\w)\right\}/R^{(o)}(\w)=o_P(1),\\
&(\text{\romannumeral5})\ \supw
T_0^{-1}(T-T_0)\left|\e^\top\mbf{M}(\w)\A(\w)\bmu_1\right|/R^{(o)}(\w)
=o_P(1).
\end{array}
$$

Thus, we just need to verify $T_0^{-1}(T-T_0)\|\Delta\|^2/R^{(o)}(\w)=o_P(1)$. We shall prove it. 
From Condition C6, we have 
\begin{align}
\sbf{\mu}_1 = & \mbf{Y}_0\mbf{W}\sbf{\theta}^* + \mbf{e} \ \ \text{and} \ \
\sbf{\mu}_1^{(o)} =  \mbf{Y}_0^{(o)}\mbf{W}\sbf{\theta}^* + \mbf{e}^{(o)}, \label{model-out-e}
\end{align}
where $\mbf{e}=(e_1,\cdots, e_{T_0})^{\top}$ and $\mbf{e}^{(o)}=(e_{T_0+1},\cdots, e_{T})^{\top}$. 
Denote 
\begin{align*}
\sbf{\Delta}_1= &\mbf{Y}_0\mbf{W}\left[\text{diag}(\mbf{Y}_0^{\top}\mbf{Y}_0)\right]^{-1}\mbf{Y}_0^{\top}\y_1-\mbf{Y}_0\mbf{W}\sbf{\theta}^*, \notag\\ 
\sbf{\Delta}_1^{(o)}= &\mbf{Y}_0^{(o)}\mbf{W}\left[\text{diag}(\mbf{Y}_0^{\top}\mbf{Y}_0)\right]^{-1}\mbf{Y}_0^{\top}\y_1-\mbf{Y}_0^{(o)}\mbf{W}\sbf{\theta}^*.
\end{align*} 
\eqref{model-out-e} follows  
\begin{align*}
\left\|\mbf{M}(\w)\y_1-\bmu_1\right\|^2 
= & \left\| \sbf{\Delta}_1 \right\|^2 +2\sbf{e}^{\top}\sbf{\Delta}_1+\|\sbf{e}\|^2,\notag\\
\left\|\mbf{M}^{(o)}(\w)\y_1-\bmu_1^{(o)}\right\|^2 
= & \left\| \sbf{\Delta}_1^{(o)} \right\|^2 +2\sbf{e}^{(o)\top}\sbf{\Delta}_1^{(o)}+\|\sbf{e}^{(o)}\|^2.
\end{align*}
Since $\|\sbf{e}\|^2$ and $\|\sbf{e}^{(o)}\|^2$ do not rely on $\w$, in order to prove the asymptotic optimality,
it suffices to verify that 
\begin{align*}
&(\text{\romannumeral6})\
\supw
T_0^{-1}(T-T_0)\left\|\sbf{\Delta}_1\right\|^2/R^{(o)}(\w)=o_P(1), \notag\\
&(\text{\romannumeral7})\
\supw
\left\|\sbf{\Delta}_1^{(o)}\right\|^2/R^{(o)}(\w)=o_P(1). 
\end{align*}
From Condition C6, we have 
$$\left[\text{diag}(\mbf{Y}_0^{\top}\mbf{Y}_0)\right]^{-1}\mbf{Y}_0^{\top}\y_1 - \sbf{\theta}^* = O_p(T_0^{-1/2}).$$
It follows that $\left\|\sbf{\Delta}_1\right\|^2=O_P(J)$ and $T_0(T-T_0)^{-1}\left\|\sbf{\Delta}_1^{(o)}\right\|^2=O_P(J)$. 
Thus, (\romannumeral6) and (\romannumeral7) are proved.
The theorem is proved.
\hfill$\square$\\

\section*{Appendix B: A Placebo Permutation Test}
\label{sec:test}

To perform inference on the estimated causal effect, 
we apply the placebo permutation based approach test \citep{ADH:10}. 
It applies the synthetic controls estimator to each control unit by pretending this control unit is the treated one. 
If there is an actual treatment effect only in the treatment group post-intervention, then the estimated effect for the actual treatment unit should be among the most extreme. 
Algorithm \ref{alg:test} provides the pseudo-code for the placebo permutation test. 
The obtained probability $p_t$ provides the probability of observing a difference
between the observable $Y_{1t}(0)$ and the estimated counterfactual $\hat{Y}_{1t}(0)$ given all permutations
of the treatment and control units.

However, this permutation test does not account for the uncertainty associated with both the unit regressing step and the estimation of the variance of the idiosyncratic shocks. 
While the estimators of coefficients in the unit regression step and the variance estimator are expected to be consistent with their true values under under the model \eqref{model:mu1-model}, controlling the impact of these uncertainties on the inference shows promise under certain conditions. 
\cite{CWZ:21} develop a permutation inference procedure for counterfactual and synthetic controls. 
However, this procedure requires consistent estimation of the counterfactual mean proxies. 
As noted in Theorem \ref{th:opt-out}, the SRC estimator is biased. 
Therefore, applying this procedure would require a debiased version of the SRC estimator, which can be constructed using projection theory \citep{ZZ:14,VB:14,Li:20}. 
Deriving such an estimator and establishing its theoretical properties, however, is beyond the scope of this paper. 

\begin{algorithm}
\begin{algorithmic}
\item Step 1: 
\For{$j=1,\ldots,J+1$}
\State Obtain the SRC estimator $\hat{Y}_{jt}(0)$ by treating unit $j$ as the treated unit;
\State Compute the differences $d_{jt}=Y_{jt}-\hat{Y}_{jt}(0)$. 
\EndFor
\item Step 2: 
\For{$t=1,\ldots,T$} 
\State Sort $d_{jt}$ decreasingly in $j$; 
\State Compute the probability $p_t$ of obtaining a value of $d_{1t}$
as $p_t=\frac{\text{rank}(d_{1t})}{J+1}$. 
\EndFor
\end{algorithmic}
\caption{The Placebo Permutation Test}
\label{alg:test}
\end{algorithm}

\section*{Appendix C: Additional Empirical Results}
\label{sec:append-B}
\begin{table}[!ht]
\caption{Post-intervention MSPE of alternative estimators under autoregressive noise.}
\begin{center}
\begin{tabular}{c | c c c c c c c c c}
\hline
$(\text{model},\sigma)$ & \scriptsize SC & \scriptsize dSC & \scriptsize ASC & \scriptsize GSC  & \scriptsize MASC & \scriptsize OLS & \scriptsize lasso & \scriptsize SRC\\
\hline
$(\emph{F1},1)$&2.116&2.062&2.108&2.365&\bf1.915& 4.074&2.084&2.438\\
$(\emph{F1},0.5)$&0.672&0.642&0.745&0.678&\bf0.512& 1.471&0.661&0.840\\
$(\emph{F1},0.1)$&0.126&0.032&0.031&0.028&\bf0.025& 0.046&0.029&0.029\\
\hline
$(\emph{F2},1)$&6.356&6.624&6.432&7.513&5.647& 5.105&6.363&\bf3.407\\
$(\emph{F2},0.5)$&5.101&5.173&5.090&5.600&6.018& 1.390&4.995&\bf0.948\\
$(\emph{F2},0.1)$&4.229&4.232&4.233&4.385&3.401& 0.060&4.198&\bf0.042\\
\hline
$(\emph{F3},1)$&6.416&6.846&6.757&7.770&7.443&12.072&6.932&\bf6.090\\
$(\emph{F3},0.5)$&5.132&5.196&5.259&5.497&5.766& 8.392&5.338&\bf3.173\\
$(\emph{F3},0.1)$&4.589&4.737&4.614&4.909&3.830& 7.618&4.760&\bf2.589\\
\hline
\end{tabular}
\end{center}
\label{simulation-ARnoise}
\end{table}

\begin{table}[!ht]
\caption{Pre-period fit of alternative estimators in the placebo study, 
measured by the the mean squared residuals.
}
\begin{center}
\begin{tabular}{c | c c c c c c c c}
\hline
region & \scriptsize SC & \scriptsize dSC & \scriptsize ASC &  \scriptsize GSC &\scriptsize MASC &  \scriptsize OLS & \scriptsize lasso & \scriptsize SMC\\
\hline
Andalucia & 1.26e-3 & 5.01e-2 & 3.96e-3 & 2.19e-4 & 2.44e-3 & 1.76e-2 & 8.17e-2 & \bf6.83e-5\\ 
Aragon & 4.12e-4 & 1.95e-4 & 1.35e-4 & 1.72e-4 & 7.52e-4 &  \bf6.46e-5 & 1.15e-3 & 6.10e-4\\ 
Asturias & 9.15e-5 & \bf4.96e-5 & 9.76e-3 & 4.26e-4 & 1.09e-4 & 1.34e-3 & 2.42e-3 & 9.14e-5\\  
Baleares & 9.52e-2 & 3.04e-1 & 4.73e-3& \bf2.78e-3  & 9.52e-2 &  5.35e-3 & 6.2e-1 & 1.11e-2\\ 
Canarias & 1.32e-3 & 2.78e-3 & 1.16e-3 & \bf4.33e-4 & 1.32e-3 &  2.17e-3 & 3.8e-2 & 6.14e-4\\ 
Cantabria & 1.09e-3 & 4.36e-4 & 4.57e-3 & \bf1.71e-4 & 5.67e-4 &  4.52e-3 & 9.89e-4 & 6.01e-4\\ 
Leon & 4.60e-4 & 8.42e-4 & \bf8.76e-5 & 1.26e-4 & 5.79e-4 & 3.18e-4 & 6.26e-3 & 2.91e-4\\  
Mancha & 4.14e-3 & 3.72e-4 & 5.96e-4 & \bf1.79e-4 & 8.4e-3 &  1.05e-3 & 9.41e-3 & 9.98e-4\\ 
Cataluna & 1.36e-2 & \bf1.83e-4 & 8.95e-4 & 3.94e-4 & 1.36e-2 & 1.39e-3 & 1.04e-1 & 4.22e-4\\
Valenciana & 9.81e-4 & 6.96e-4 & 7.42e-4 & 9.47e-4 & 4.51e-3 & \bf1.45e-4 & 2.48e-3 & 8.16e-4\\
Extremadura & 1.15e-1 & 1.08e-2 & 1.40e-3 & \bf1.07e-3 & 1.15e-1 &  2.04e-3 & 8.53e-2 & 1.30e-2\\
Galicia & 3.01e-4 & 3.45e-4 & \bf3.07e-5 & 1.44e-4 & 3.01e-4 &  8.51e-5 & 1.22e-2 & 5.30e-4\\ 
Madrid & 7.21e-1 & 5.61e-2 & 2.25e-2 & \bf8.76e-3 & 7.21e-1 & 2.14e-2 & 1.19e-1 & 2.08e-1\\ 
Murcia & 1.37e-3 & 5.33e-4 & 1.44e-3 & \bf3.11e-4 & 2.63e-3 &  3.11e-3 & 9.07e-3 & 6.42e-3\\ 
Navarra & 2.82e-4 & 3.01e-4 & \bf6.16e-5 & 2.79e-4 & 5.09e-4 & 1.63e-4 & 1.13e-2 & 5.14e-4\\ 
Rioja & 7.17e-4 & 8.40e-4 & \bf1.30e-4 & 4.36e-4 & 3.11e-3 & 4.95e-4 & 1.38e-2 & 3.98e-4\\
\hline
average & 5.98e-2 & 2.68e-2 & 3.27e-3 & \bf1.05e-3 & 6.06e-2 & 3.83e-2 & 6.98e-2 & 1.53e-2\\
\hline
\end{tabular}
\end{center}
\label{realdata-fit}
\end{table}

\end{document}